\documentclass[aps,pra,amsmath,twocolumn,showpacs,amssymb,superscriptaddress]{revtex4}

\usepackage{CJK}
\usepackage{algpseudocode}
\usepackage{algorithm}
\usepackage[dvips]{graphicx}
\usepackage{fancyhdr}
\usepackage[nolabel]{showlabels}
\usepackage{amssymb}
\usepackage[all]{xy}
\usepackage[colorlinks=true, urlcolor=rltblue, citecolor=drkgreen] {hyperref}
\usepackage{color}
\usepackage[mathscr]{eucal}
\usepackage{pdfsync}
\definecolor{rltblue}{rgb}{0,0,0.4}
\definecolor{drkgreen}{rgb}{0,0.4,0}
\usepackage{color}
\usepackage{times}
\usepackage{bm}
\usepackage{dcolumn}% Align table columns on decimal point
\usepackage{bm}% bold math
\usepackage[english]{babel}
\usepackage{amsthm}
\usepackage{amsmath}
\usepackage{algpseudocode}

\newtheorem{thm}{Theorem}
\newtheorem{lemma}[thm]{Lemma}
\newtheorem*{lemma*}{Lemma}

\theoremstyle{definition}
\newtheorem{definition}[thm]{Definition}

\theoremstyle{remark}

\newtheorem{historic}[thm]{Historic Remark}

\theoremstyle{plain}

\newcounter{contenumi}

\def\and{\mathrel{\&}}

\def\ZZ{{\mathbb Z}}

\def\CC{{\mathbb C}}
\def\DD{\mathbb D}
\def\EE{{\mathbb E}}
\def\VV{{\mathbb V}}
\def\FF{\mathbb F}

\def\MM{{\mathbb M}}

\def\RR{{\mathbb R}}

\begin{document}
	
	\title{The spark of synchronization in heterogeneous networks of chaotic maps}
	
	\author{Antonio Montalb\'an}
	\affiliation{Department of Mathematics, University of California, Berkeley, Berkeley 94720, USA} 
	\author{Rodrigo M. Corder}
	\email{rodrigo.corder@usp.br}
	\affiliation{Divisions of Epidemiology and Biostatistics, School of Public Health, University of California, Berkeley, Berkeley 94720, USA}
	\affiliation{Department of Parasitology, Institute of Biomedical Science, University of São Paulo, São Paulo 05508-000, Brazil}
	
	\date{\today}
	
	\begin{abstract}
		We investigate the emergence of synchronization in heterogeneous networks of chaotic maps. Our findings reveal that a small cluster of highly connected maps is responsible for triggering the spark of synchronization. After the spark, the synchronized cluster grows in size and progressively moves to less connected maps, eventually reaching a cluster that may remain synchronized over time. We explore how the shape of the network degree distribution affects the onset of synchronization and derive an expression based on the network construction that determines the expected time for a network to synchronize. Understanding how the network design affects the spark of synchronization is particularly important for the control and design of more robust systems that require some level of coherence between a subset of units for better functioning. Numerical simulations in finite-sized networks are consistent with this analysis.
	\end{abstract}
	
	\pacs{05.45.Xt, 89.75.Hc, 05.45.Ac}
	\maketitle

	%%%%%%%%%%%%%%%%%%%%%%%%%%%%%%%%%%%%%%%%%%%%%%%%%%%%%%%%%%%%%%%%%%%%%%%%%%%%%%%%%%%%%%%%%%%%%%%%%%%%%%%%%%%%%%%%%%%%%%%%%%%%%%%%%%%%%%%%%%%%%%%%%%%%%%%%%%%%%%%%%%%%%%%%%%%%%%%%%%%%%%%%%%%%%%%%%%%%%%%%%%%%%%%%%%%%%%%%%%%%%%%%%%%%%%%%%%%%%%%%%%%%%%%%%%%%%%%%%%%%%%%%%%%%%%%%%%%%%%%%%%%%%%%%%%%%%%%%%%%%%%%%%%%%%%%%%%%%%%%%%%%%%%%%%%%%%%%%%%%%%%%%%%%%%%%%%%%%%%%%%%%%%%%%%%%%%%%%%%%%%%%%%%%%%%%%%%%%%%%%%%%%%%%%%%%%%%%%%%%%%%%%%%%%%%%%%%%%%%%%%%%%%%%%%%%%%%%%%%%%%%%%%%%%%%%%%%%%%%%%%%%%%%%%%%%%%%%%%%%%%%%%%%%%%%%%%%%%%%%%%%%%%%%%%%%%%%%%%%%%%%%%%%%%%%%%%%%%%%%%%%%%%%%%%%%%%%%%%%%%%%%%%%%%%%%%%%%%%%%%%%%%%%%%%%%%%%%%%%%%%%%%%%%%%%%%%%%%%%%%%%%%%%%%%%%%%%%%%%%%%%%%%%%%%%%%%%%%%%%%%%%%%%%%%%%%%%%%%%%%%%%%%%%%%%%%%%%%%%%%%%%%%%%%%%%%%%%%%%%%%%%%%%%%%%%%%%%%%%%%%%%%%%%%%%%%%%%%%%%%%%%%%%%%%%%%%%%%%%%%%%%%%%%%%%%%%%%%%%%%%%%%%%%%%%%%%%%%%%%%%%%%%%%%%%%%%%%%%%%%%%%%%%%%%%%%%%%%%%%%%%%%%%%%%%%%%%%%%%%%%%%%%%%%%%
	{\bf   Synchronization stands as a pivotal phenomenon in networks, exerting a profound impact across a spectrum of disciplines, including biology, chemistry and physics, and various man-made systems. Notably, recent investigations unveiled the phenomenon of chaotic units achieving sustained and stable cluster synchronization within heterogeneous networks. Because many real-world systems rely on cluster synchronization for their functioning and natural systems often comprise individuals with varying connection counts, understanding the fundamental mechanisms that underlie the spark of synchrony in heterogeneous networks holds promise for elucidating innovative control strategies that amplify coherence among interacting entities. While extensive research has addressed the conditions for synchronization, the dynamics preceding the spark of synchrony remains elusive.}
	
	\section{Introduction} \label{Sec:Introduction}
	
	The complex structure of real-world networks has been extensively studied, with significant progress made in recent years \cite{Milgram1967, Barabasi1999, Newman2002, Radicchi2004, Zhou2004}. These networks encompass a wide range of dynamical processes, from natural systems in biology \cite{Barabasi2004} to man-made technological systems \cite{Pastor2001}. In particular, when a network is composed of interacting dynamical systems, it can exhibit large-scale coherent behavior that spontaneously emerges under certain conditions \cite{Arenas2008}. For instance, fireflies in a swarm are known to synchronize their rhythms of flashing, resulting in highly correlated flashes among a significant proportion of the swarm \cite{Sarfati2023}. This phenomenon, known as synchronization, has also been observed in various other systems \cite{Winfree2002}.
	
	Researchers have conducted extensive investigations into a class of models called the Kuramoto models, which aim to describe the dynamics of coupled oscillators \cite{Strogatz2004, Rodrigues2016}. Kuramoto himself initially investigated the case of fully connected networks with $N$ coupled phase oscillators of equal strength \cite{Kuramoto2013}. His work revealed that in the continuum limit, there exists a critical coupling strength value---dependent on the distribution of the phase oscillator frequencies---that determines whether the phases of the oscillators in a coupled network will remain incoherent or eventually evolve into synchronized behavior. Later, it was found that the network topology has a significant influence on the dynamics of these systems \cite{Lee2005, Ichinomiya2004, GomesGardenes2007}.
	
	Extensive research has also been conducted on coupled collections of different systems with more general dynamics, such as mixed chaotic and periodic oscillators, chaotic maps, and others \cite{Baek2004, Viana2005, Pereira2017}. In particular, in collaboration with Tiago Pereira and Zheng Bian \cite{Corder2023}, we unravel the mechanism for the emergence of cluster synchronization in heterogeneous random networks. We developed a heterogeneous mean field approximation together with a self-consistent theory to determine the onset and stability of the cluster. The analysis showed that cluster synchronization occurs in a wide variety of heterogeneous networks. The system dynamics before its asymptotic behavior, however, remains undisclosed in random networks.
	
	%In particular, the study of heterogeneously coupled maps \cite{Pereira2017} addresses the class of chaotic problems and incorporates nonlinear and high-dimensional behavior observed in many networks. The emergence of coherence in large collections of heterogeneous coupled chaotic systems is also highly dependent on the network topology. It was found that, for these systems, the emergent dynamics changes across the network connectivity levels \cite{Pereira2017}. Additionally, it was shown that the critical coupling strengths are described by both the network structure and the individual unit dynamics \cite{Restrepo2006}.
	
	%The aim of this work is to investigate the emergence of synchronization in chaotic maps from a rather new perspective. We investigate the expected time until the emergence of coherence and, as a study case, we consider heterogeneous networks of individuals whose dynamics are governed by Bernoulli maps \cite{Liang2022}. Notably, we do not assume any fixed form for the network degree distribution. For simplicity, we consider interactions among coupled maps that can be described by a sinusoidal function \cite{Stankovski2017}. We study the network dynamics before synchrony and depict fundamental processes that drive the spark of spontaneous synchronization. Furthermore, we present numerical simulations in finite-sized networks consistent with the theoretical analysis. 
	
	The transition to synchronization in coupled networks  was studied from different perspectives, and it was observed that this phenomenon either occurs or is inhibited due to finite size fluctuations of the dynamical systems \cite{Komarov2015, Ottino2018}. In particular, it was observed that the transition time to synchrony is exponentially distributed in large homogeneous networks of chaotic circle maps \cite{Mendonca2023}. Here, we study the network dynamics before synchrony considering heterogeneous networks of individuals whose dynamics is governed by Bernoulli maps \cite{Liang2022}. Notably, we do not assume any fixed form for the network degree distribution and, for simplicity, we consider interactions among coupled maps that can be described by a sinusoidal function \cite{Stankovski2017}. We depict fundamental processes that drive the spark of spontaneous synchronization and observe that the transition time to synchrony is also exponentially distributed in heterogeneous networks of chaotic maps. We show that the theoretical analysis is consistent with numerical simulations in finite-sized networks.
	
	%This work is organized as follows: Section \ref{Sec:Example} illustrates the phenomenon under investigation with a computational simulation, highlighting the crucial role of the connected oscillators in sparking synchronization, and introduces the question we investigate through the paper. In Section \ref{Sec:The model}, we define the dynamics of the system, the network structure and define some parameters that will be used to measure the levels of synchronization in the network. We start discussing the phenomenon of the emergence of synchrony in Section \ref{Sec:The initiation of Synchronization} and, in Section \ref{Sec:The complex-square+noise Markov chain}, we discuss the general problem of the spark of synchronization. In Section \ref{Sec:The main result}, we introduce our main result about the expected time for this event to occur given a network structure. In Section \ref{Sec:The continuous approximation}, we present the continuous approximation, derive the formulas presented in Section \ref{Sec:The main result}, and investigate how the shape of the degree distribution affects the emergence of synchrony. In the last section we discuss our results.
	
	This work is organized as follows: in Sec. \ref{Sec:The model}, we define the dynamics of the system, the network structure, and some parameters that will be used to measure the levels of synchronization among units. Section \ref{Sec:Example} illustrates the phenomenon under investigation with a computational simulation, highlights the crucial role of the connected oscillators in sparking synchronization, and introduces the question we investigate through the paper. We start discussing the phenomenon of the emergence of synchrony in Sec. \ref{Sec:The initiation of Synchronization}, and in Sec. \ref{Sec:The complex-square+noise Markov chain}, we introduce the general model that describes the spark of synchronization. In Sec. \ref{Sec:The main result}, we apply the model to describe the spark of synchronization given a network structure. In Sec. \ref{Sec:The continuous approximation}, we present the continuous approximation of the model, derive the formulas presented in Sec. \ref{Sec:The main result}, and investigate how the shape of the degree distribution affects the emergence of synchrony. In the last section, we discuss our results.

	%%%%%%%%%%%%%%%%%%%%%%%%%%%%%%%%%%%%%%%%%%%%%%%%%%%%%%%%%%%%%%%%%%%%%%%%%%%%%%%%%%%%%%%%%%%%%%%%%%%%%%%%%%%%%%%%%%%%%%%%%%%%%%%%%%%%%%%%%%%%%%%%%%%%%%%%%%%%%%%%%%%%%%%%%%%%%%%%%%%%%%%%%%%%%%%%%%%%%%%%%%%%%%%%%%%%%%%%%%%%%%%%%%%%%%%%%%%%%%%%%%%%%%%%%%%%%%%%%%%%%%%%%%%%%%%%%%%%%%%%%%%%%%%%%%%%%%%%%%%%%%%%%%%%%%%%%%%%%%%%%%%%%%%%%%%%%%%%%%%%%%%%%%%%%%%%%%%%%%%%%%%%%%%%%%%%%%%%%%%%%%%%%%%%%%%%%%%%%%%%%%%%%%%%%%%%%%%%%%%%%%%%%%%%%%%%%%%%%%%%%%%%%%%%%%%%%%%%%%%%%%%%%%%%%%%%%%%%%%%%%%%%%%%%%%%%%%%%%%%%%%%%%%%%%%%%%%%%%%%%%%%%%%%%%%%%%%%%%%%%%%%%%%%%%%%%%%%%%%%%%%%%%%%%%%%%%%%%%%%%%%%%%%%%%%%%%%%%%%%%%%%%%%%%%%%%%%%%%%%%%%%%%%%%%%%%%%%%%%%%%%%%%%%%%%%%%%%%%%%%%%%%%%%%%%%%%%%%%%%%%%%%%%%%%%%%%%%%%%%%%%%%%%%%%%%%%%%%%%%%%%%%%%%%%%%%%%%%%%%%%%%%%%%%%%%%%%%%%%%%%%%%%%%%%%%%%%%%%%%%%%%%%%%%%%%%%%%%%%%%%%%%%%%%%%%%%%%%%%%%%%%%%%%%%%%%%%%%%%%%%%%%%%%%%%%%%%%%%%%%%%%%%%%%%%%%%%%%%%%%%%%%%%%%%%%%%%%%%%%%%%%%%%%%%%
	\section{The model} \label{Sec:The model}
	%\subsection{The model}
	%{\bf The model.} \label{Sec:The model}
	We study networks with $N$ coupled maps $z_i$ satisfying 
	\begin{equation}\label{Eq:M0}
		z^{t+1}_i \quad = \quad  2 z^{t}_i + \frac{\alpha}{C} \sum_{j}  A_{ij}\sin(z^{t}_j - z^{t}_i) \quad\quad \mbox{mod $2\pi$}.
	\end{equation}
	Here, $i=1,2,...,N$ labels each map, $\alpha$ is the {\em network coupling strength}, $C$ is the {\em network mean degree}, $A_{ij}$ is 1 if nodes $i$ and $j$ are connected and 0 otherwise, and $z^{t}_i \in \RR / 2\pi\ZZ$ is the state of map $i$ at time $t$. We use $S^1$ to denote $\RR / 2\pi\ZZ$.
	
	Let $\delta:\RR^+\to \RR^+$ be a probability density function with mean 1. We consider random networks with {\em degree distributions} $d \mapsto \delta(d/C)$ and denote by $d_i$ the degree of node $i$. 
	%In this framework, the number of nodes $i$ with degree $d$ is $\delta(d/C)N/C$.
	Let $w_i= d_i/C$. Thus, $w_i$ follows the distribution $\delta$ with mean 1. Notice that because we consider random graphs, the probability that node $i$ is connected to node $j$ is $d_i d_j/CN  = w_iw_j C /N$.

	%%%%%%%%%%%%%%%%%%%%%%%%%%%%%%%%%%%%%%%%%%%%%%%%%%%%%%%%%%%%%%%%%%%%%%%%%%%%%%%%%%%%%%%%%%%%%%%%%%%%%%%%%%%%%%%%%%%%%%%%%%%%%%%%%%%%%%%%%%%%%%%%%%%%%%%%%%%%%%%%%%%%%%%%%%%%%%%%%%%%%%%%%%%%%%%%%%%%%%%%%%%%%%%%%%%%%%%%%%%%%%%%%%%%%%%%%%%%%%%%%%%%%%%%%%%%%%%%%%%%%%%%%%%%%%%%%%%%%%%%%%%%%%%%%%%%%%%%%%%%%%%%%%%%%%%%%%%%%%%%%%%%%%%%%%%%%%%%%%%%%%%%%%%%%%%%%%%%%%%%%%%%%%%%%%%%%%%%%%%%%%%%%%%%%%%%%%%%%%%%%%%%%%%%%%%%%%%%%%%%%%%%%%%%%%%%%%%%%%%%%%%%%%%%%%%%%%%%%%%%%%%%%%%%%%%%%%%%%%%%%%%%%%%%%%%%%%%%%%%%%%%%%%%%%%%%%%%%%%%%%%%%%%%%%%%%%%%%%%%%%%%%%%%%%%%%%%%%%%%%%%%%%%%%%%%%%%%%%%%%%%%%%%%%%%%%%%%%%%%%%%%%%%%%%%%%%%%%%%%%%%%%%%%%%%%%%%%%%%%%%%%%%%%%%%%%%%%%%%%%%%%%%%%%%%%%%%%%%%%%%%%%%%%%%%%%%%%%%%%%%%%%%%%%%%%%%%%%%%%%%%%%%%%%%%%%%%%%%%%%%%%%%%%%%%%%%%%%%%%%%%%%%%%%%%%%%%%%%%%%%%%%%%%%%%%%%%%%%%%%%%%%%%%%%%%%%%%%%%%%%%%%%%%%%%%%%%%%%%%%%%%%%%%%%%%%%%%%%%%%%%%%%%%%%%%%%%%%%%%%%%%%%%%%%%%%%%%%%%%%%%%%%%%%%%
	%\section{The weighted order parameter $V^t$} \label{Sec:The weighted order parameter}
	We associate a complex number $u^t_i = e^{i z^t_i}$ to each state $z^t_i$. In this paper, we will move interchangeably between a state represented as a real number $z^t_i\in \RR/2\pi\ZZ$ or a complex number $u^t_i$ in the unit circle. For each map $i$, we define $\vec{V}^t_i\in \mathbb{C}$, $r^t_i \in \RR^+$, and $\theta^t_i \in [0,2\pi]$ as follows:
	\begin{equation}
		V^t_i  = \sum_{j} A_{ij} u^t_j = r^t_i e^{i \theta^t_i}.
		\label{Eq:Vi}
	\end{equation}
	
	By writing $\sin(z^{t}_j - z^{t}_i)$ as the imaginary part of $u^t_j\bar{u}^i_j$ (where $\bar{u}$ is the complex conjugate of $u$), using Eq. \ref{Eq:Vi}, we rewrite Eq. \ref{Eq:M0} as
	\begin{equation}
		z^{t+1}_i = 2 z^t_i + \frac{\alpha}{C} \Im(V^t_i\bar{u}^t_i),
		\label{Eq:M1}
	\end{equation}
	where $\Im(v)$ denotes the imaginary part of the complex number $v$.

	%%%%%%%%%%%%%%%%%%%%%%%%%%%%%%%%%%%%%%%%%%%%%%%%%%%%%%%%%%%%%%%%%%
	%\subsection{The distribution of the $V^t_i$'s}
	{\bf The distribution of $V^t_i$'s.}
	We approximate $V^t_i$ with a 2 D Gaussian distribution for all maps $i$. Note that each $V^t_i$ is a sum of $d_i$ vectors of the form $u^t_j$ (for $j$ connected to $i$), which are distributed in a somewhat random manner around the circle. To calculate the mean vector of those $u^t_j$'s, we sum all states of maps $j$, each multiplied by the probability that map $j$ is connected to $i$, namely, $w_i w_j C/N$, and divide by the total number of maps $d_i=C w_i$ connected to $i$. 
	We obtain a vector that we call {\em the weighted order parameter} and denote by $V^t$, i.e.,
	\begin{equation}
		V^t = \frac{1}{N}\sum_{j} w_j \cdot u^t_j.
		\label{Eq:V}
	\end{equation}
	Notice that $V^t$ is independent of $i$ and that the expected value for $V^t_i$ is $d_i\cdot V^t$.
	
	Next, to calculate the covariance matrix of $u^t_j$'s, denoted by $\Sigma$ (Eq. \ref{Eq:Sigma}), we consider the case when they are uniformly distributed in the circle (the worst case scenario). In this framework, on each coordinate, the variance of $\sin(z)$ for $z$ uniformly in the interval $[0,2\pi]$ is $1/2$. The covariance matrix $\Sigma$ is, thus, written as
	\begin{equation}
		\Sigma = \frac{1}{2}
		\begin{pmatrix}
			1 & 0 \\
			0 & 1
		\end{pmatrix}.
		\label{Eq:Sigma}
	\end{equation}
	Therefore, by the central-limit theorem, we can estimate $V^t_i \sim {\mathcal N}^2(d_i V^t, (d_i/2) {\bf I})$, where ${\mathcal N}^2$ is the Gaussian distribution on the plane and $\bf{I}$ is the 2x2 identity matrix.

	%%%%%%%%%%%%%%%%%%%%%%%%%%%%%%%%%%%%%%%%%%%%%%%%%%%%%%%%%%%%%%%%%%
	%\subsection{The weighted order parameter $V^t$}
	
	{\bf The weighted order parameter $V^t$.}
	The weighted order parameter $V^t$ is the average of $V^t_i$'s divided by the network mean degree $C$, i.e., $V^t = \frac{1}{CN}\sum_{i}\sum_{j} A_{ij} u^t_i$. It is computed by taking into account the degree of each map $j$. Following the same procedure used to define $V^t_i$, we can define $r^t \in \RR^+$ and $\theta^t \in [0,2\pi]$ such that $V^t=r^t e^{i\theta^t}$. Notice that the most common approach in the literature is to compute the order parameter as $\frac{1}{N}\sum_{j} u^t_j$, i.e., without taking into account the network degree distribution. In our heterogeneous framework, maps with higher connectivity exert a greater influence on network synchronization compared to poorly connected maps. The weighted order parameter $V^t$ (Eq. \ref{Eq:V}) better captures the phenomenon of synchronization in heterogeneous networks. This is because only a cluster of maps with similar connectivity synchronizes at each time step. When such clusters are small, little variation is observed in the unweighted order parameter. In particular, this is more evident at the moment of spark, when only a few highly connected maps spontaneously synchronize, making the weighted order parameter more appropriate to capture this event.

	%%%%%%%%%%%%%%%%%%%%%%%%%%%%%%%%%%%%%%%%%%%%%%%%%%%%%%%%%%%%%%%%%%%%%%%%%%%%%%%%%%%%%%%%%%%%%%%%%%%%%%%%%%%%%%%%%%%%%%%%%%%%%%%%%%%%%%%%%%%%%%%%%%%%%%%%%%%%%%%%%%%%%%%%%%%%%%%%%%%%%%%%%%%%%%%%%%%%%%%%%%%%%%%%%%%%%%%%%%%%%%%%%%%%%%%%%%%%%%%%%%%%%%%%%%%%%%%%%%%%%%%%%%%%%%%%%%%%%%%%%%%%%%%%%%%%%%%%%%%%%%%%%%%%%%%%%%%%%%%%%%%%%%%%%%%%%%%%%%%%%%%%%%%%%%%%%%%%%%%%%%%%%%%%%%%%%%%%%%%%%%%%%%%%%%%%%%%%%%%%%%%%%%%%%%%%%%%%%%%%%%%%%%%%%%%%%%%%%%%%%%%%%%%%%%%%%%%%%%%%%%%%%%%%%%%%%%%%%%%%%%%%%%%%%%%%%%%%%%%%%%%%%%%%%%%%%%%%%%%%%%%%%%%%%%%%%%%%%%%%%%%%%%%%%%%%%%%%%%%%%%%%%%%%%%%%%%%%%%%%%%%%%%%%%%%%%%%%%%%%%%%%%%%%%%%%%%%%%%%%%%%%%%%%%%%%%%%%%%%%%%%%%%%%%%%%%%%%%%%%%%%%%%%%%%%%%%%%%%%%%%%%%%%%%%%%%%%%%%%%%%%%%%%%%%%%%%%%%%%%%%%%%%%%%%%%%%%%%%%%%%%%%%%%%%%%%%%%%%%%%%%%%%%%%%%%%%%%%%%%%%%%%%%%%%%%%%%%%%%%%%%%%%%%%%%%%%%%%%%%%%%%%%%%%%%%%%%%%%%%%%%%%%%%%%%%%%%%%%%%%%%%%%%%%%%%%%%%%%%%%%%%%%%%%%%%%%%%%%%%%%%%%%%%%%
	\section{The phenomenon under investigation} \label{Sec:Example}
	Figure \ref{Fig:Figure1}A presents a simulation exemplifying the phenomenon under investigation in this work. We consider a heterogeneous network with $N=100,000$ coupled maps with degree distribution $d$ randomly sampled from an inverse-gamma distribution with mean $C=1,000$, a power-law exponent $\gamma=3$, and coupling strength $\alpha=15$. The initial states of maps are randomly sampled from a uniform distribution, and the dynamics of the system is governed by Eq. \ref{Eq:M0}. In this example, we observe the emergence of spontaneous synchronization initially among few highly connected maps, which occurs at a certain moment $t<10$. This moment marks the \textit{sparks} of synchronization. Notice in Fig. \ref{Fig:Figure1}A that the cluster of synchrony progressively evolves over time, such that at each subsequent time step, the most connected maps within the cluster lose synchronization concomitantly with the synchronization of less connected maps. At a certain point, the system reaches an asymptotic behavior where only maps with similar and intermediate connectivity levels remain synchronized. Interestingly, throughout this process, the maps within the cluster of synchrony exhibit a certain level of coherence among themselves, but the cluster does not synchronize to a fixed value. We refer to this phenomenon as \textit{partial synchronization}---for a more detailed discussion of the asymptotic behavior, see Ref. \cite{Corder2023}.
	
	\begin{figure}[h]
		\centering
		\includegraphics[width=86mm]{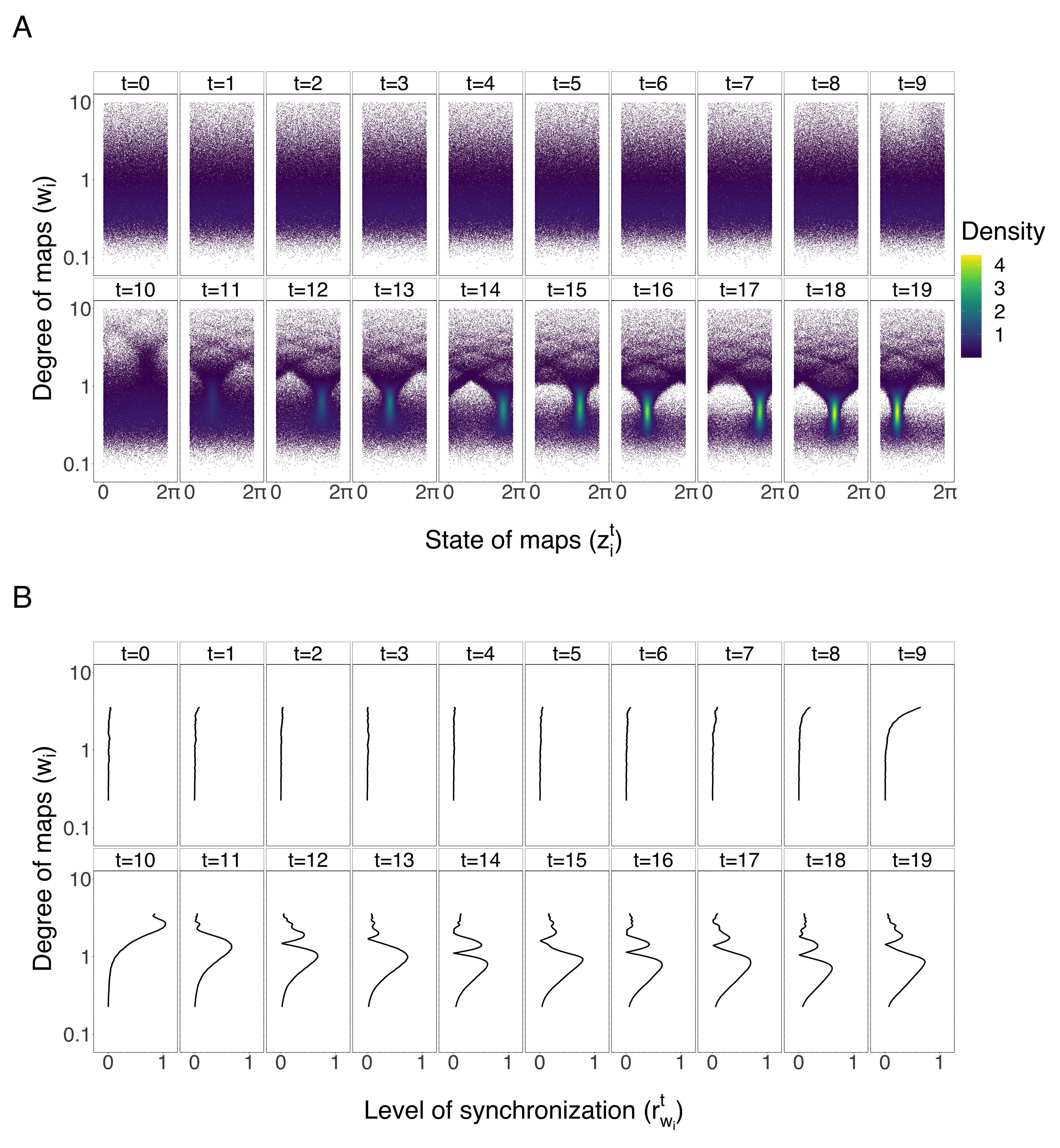}
		\caption{Emergence of partial synchronization in a heterogeneous network of coupled maps. To illustrate this phenomenon we consider a heterogeneous network with $N=100,000$ coupled maps and the degree distribution $d$ randomly sampled from an inverse-gamma distribution with mean $C=1,000$, a power-law exponent $\gamma=3$ (truncated at 10,000), and coupling strength $\alpha=15$. The figure presents a simulation where initial states of the maps ($z^{0}_i$, $i=1,2,...,N$) are randomly sampled from an uniform distribution at the beginning of the simulation, and the dynamics of the system is governed by Eq. \ref{Eq:M0}. Panel (A) shows the state of maps $z^{t}_i$ on the $x$-axis and the normalized degree $w_i=d_i/C$ on the $y$-axis, with the highest density areas in yellow. Panel (B) shows the level of synchronization on the $x$-axis, measured by the weighted order parameter $r^t_w$ (defined in Sec. \ref{Sec:The initiation of Synchronization}), and the normalized degree layer $w_i$ on the $y$-axis. We observe the emergence of spontaneous synchronization initially among few highly connected maps at a certain moment when $t < 10$. The cluster of synchrony progressively evolves over time, such that at each subsequent time step, the most connected maps within the cluster lose synchronization concomitantly with the synchronization of less connected maps. At a certain point, the system reaches an asymptotic behavior where only maps with similar and intermediate degree levels remain synchronized. Throughout this process, maps within the cluster of synchrony exhibit a certain level of coherence among themselves, but the cluster does not synchronize to a fixed value. For a more detailed discussion of the asymptotic behavior of this model, see Ref. \cite{Corder2023}.
			\label{Fig:Figure1}}
	\end{figure}

	%%%%%%%%%%%%%%%%%%%%%%%%%%%%%%%%%%%%%%%%%%%%%%%%%%%%%%%%%%%%%%%%%%%%%%%%%%%%%%%%%%%%%%%%%%%%%%%%%%%%%%%%%%%%%%%%%%%%%%%%%%%%%%%%%%%%%%%%%%%%%%%%%%%%%%%%%%%%%%%%%%%%%%%%%%%%%%%%%%%%%%%%%%%%%%%%%%%%%%%%%%%%%%%%%%%%%%%%%%%%%%%%%%%%%%%%%%%%%%%%%%%%%%%%%%%%%%%%%%%%%%%%%%%%%%%%%%%%%%%%%%%%%%%%%%%%%%%%%%%%%%%%%%%%%%%%%%%%%%%%%%%%%%%%%%%%%%%%%%%%%%%%%%%%%%%%%%%%%%%%%%%%%%%%%%%%%%%%%%%%%%%%%%%%%%%%%%%%%%%%%%%%%%%%%%%%%%%%%%%%%%%%%%%%%%%%%%%%%%%%%%%%%%%%%%%%%%%%%%%%%%%%%%%%%%%%%%%%%%%%%%%%%%%%%%%%%%%%%%%%%%%%%%%%%%%%%%%%%%%%%%%%%%%%%%%%%%%%%%%%%%%%%%%%%%%%%%%%%%%%%%%%%%%%%%%%%%%%%%%%%%%%%%%%%%%%%%%%%%%%%%%%%%%%%%%%%%%%%%%%%%%%%%%%%%%%%%%%%%%%%%%%%%%%%%%%%%%%%%%%%%%%%%%%%%%%%%%%%%%%%%%%%%%%%%%%%%%%%%%%%%%%%%%%%%%%%%%%%%%%%%%%%%%%%%%%%%%%%%%%%%%%%%%%%%%%%%%%%%%%%%%%%%%%%%%%%%%%%%%%%%%%%%%%%%%%%%%%%%%%%%%%%%%%%%%%%%%%%%%%%%%%%%%%%%%%%%%%%%%%%%%%%%%%%%%%%%%%%%%%%%%%%%%%%%%%%%%%%%%%%%%%%%%%%%%%%%%%%%%%%%%%%%%%%%
	{\bf The question under investigation.} \label{Sec:Question}
	Given a network design (defined by parameters $N,C,\alpha, \delta$), what is the expected time for the spark to occur when the network dynamics is governed by Eq. \ref{Eq:M0}? To answer this question, we formulate a single number $s(N,C,\alpha, \delta)$ [Eq. \ref{Eq:s}], which encapsulates all information about the network structure that allows us to estimate the expected time for the spark to occur. The number $s(N,C,\alpha, \delta)$ is defined, as we will see in detail in Sec. \ref{Sec:The main result}, as follows:
	\begin{equation}
		s(N,C,\alpha, \delta) = 
		\frac{\alpha^2}{\sqrt{2N}}\  M^\alpha(\delta) \  K^\alpha_\delta(C)/8,
		\label{Eq:s}
	\end{equation}
	where
	\begin{equation}
		M^\alpha(\delta)= M_3(\delta) \left(\sqrt{M_2(\delta)} + \frac{4}{\alpha\sqrt{M_2(\delta)}}\right),
		\label{Eq:M}
	\end{equation}
	with $M_2(\delta)$ and $M_3(\delta)$ representing, respectively, the second and third cumulative moments of the probability density function $\delta$. Note that $M^\alpha(\delta)$ depends solely on $\alpha$ and $\delta$, and it measures the effect of the network degree distribution sparsity on the expected time to spark. The term
	\begin{equation}
		K^\alpha_\delta(C) = \frac{\int_0^{\infty} w^3 e^{-w \frac{\alpha^2}{4C}} \delta(w) dw}{\int_0^{\infty} w^3 \delta(w) dw}
		\label{Eq:K}
	\end{equation}
	measures the impact of the noise resulting from the fact that each map is connected to a distinct set of neighbors---the term $K^\alpha_\delta(C)$ is, therefore, negligible for large values of $C$.
	
	We will show in Sec. \ref{Sec:The main result} that the expected time to the spark of synchronization is exponentially distributed and can be approximated by the exponential of multiples of $s(N,C,\alpha, \delta)^{-2}$. An example of the exponential behavior of the expected time to spark is presented in Fig. \ref{Fig:Figure2}A, where simulations were performed for different networks designs with degrees randomly sampled from inverse-gamma distributions. The algorithm to calculate the expected time to the emergence of partial synchrony in heterogeneous network of coupled maps is presented in  Appendix \ref{Sec:Appendix Simulation}.
	
	\begin{figure}[h]
		\centering
		\includegraphics[width=86mm]{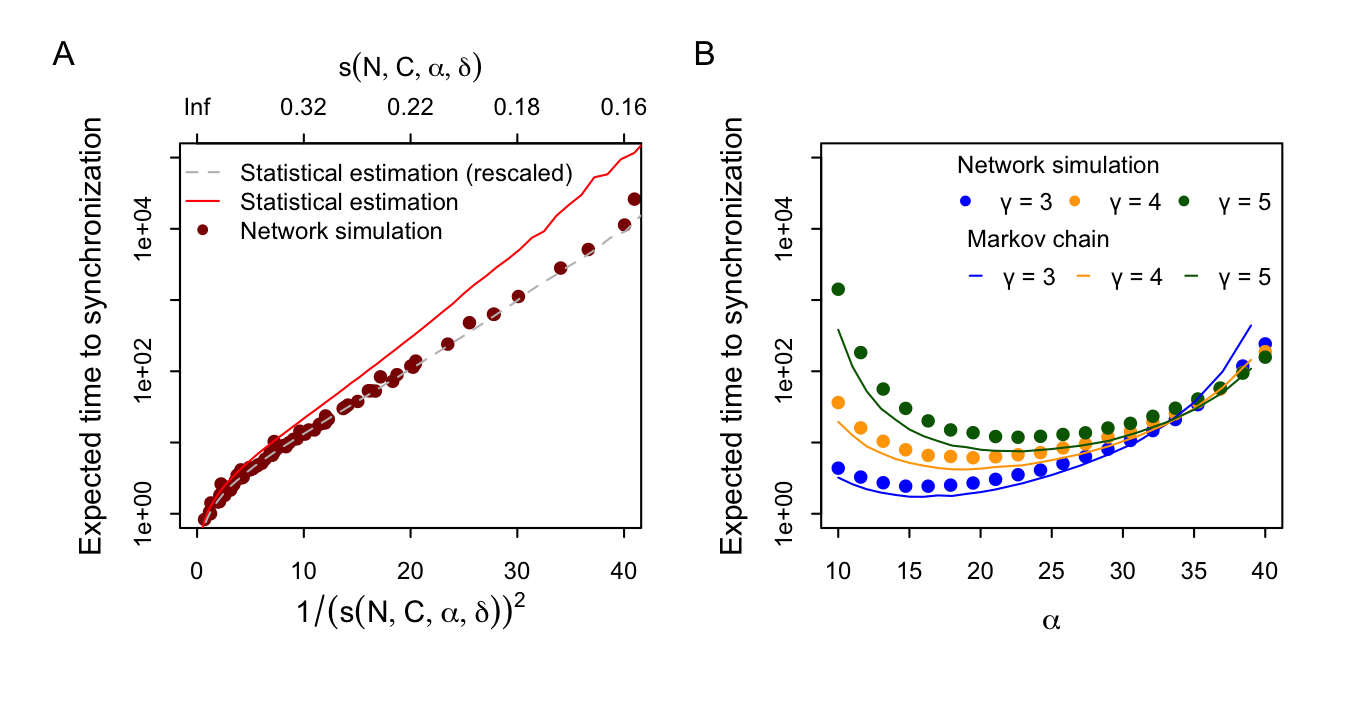}
		\caption{Expected time for the emergence of partial synchronization in a heterogeneous network of coupled maps. Panel (A) presents the expected sparking time of partial synchronization (in log-scale) as function of $s(N,C,\alpha, \delta)^{-2}$ [Eq. \ref{Eq:s}]. Each dot corresponds to the expected time for the emergence of partial synchronization considering different network configurations: $N\in\{30000,50000,100000\}$, $C\in\{300, 500, 1000\}$, $\alpha\in\{10,12,15\}$, and $\gamma = \{3,4,5\}$ (power-law exponent of inverse-gamma distribution $\delta$ with mean 1 and truncated at $10$---re-scaled to keep the mean 1). The solid red line shows the expected sparking time given by our theoretical results [$\EE(\sigma)$, as defined in Sec. \ref{Sec:The complex-square+noise Markov chain}]. The log and inverse-square scales were chosen to highlight the proportionality between coordinates, showing that the expected time to spark behaves as an exponential of $s(N,C,\alpha, \delta)^{-2}$. The dashed gray line shows the expected sparking time given by the theoretical results with $\sigma$ re-scaled by 10\% [that is, $\EE(1.1\cdot\sigma)$; more details in Sec. \ref{Sec:The complex-square+noise Markov chain}]. Panel (B) shows the expected time to spark for a network with $N=30,000$ coupled maps and mean degree $C=300$ computed for different values of $\alpha$ and $\gamma$. Dots are obtained from network simulations and solid lines are the statistical estimations (as described in Sec. \ref{Sec:The complex-square+noise Markov chain}). Blue, orange, and green colors represent $\gamma=3,4$ and $5$, respectively.
			\label{Fig:Figure2}}
	\end{figure}

	%%%%%%%%%%%%%%%%%%%%%%%%%%%%%%%%%%%%%%%%%%%%%%%%%%%%%%%%%%%%%%%%%%%%%%%%%%%%%%%%%%%%%%%%%%%%%%%%%%%%%%%%%%%%%%%%%%%%%%%%%%%%%%%%%%%%%%%%%%%%%%%%%%%%%%%%%%%%%%%%%%%%%%%%%%%%%%%%%%%%%%%%%%%%%%%%%%%%%%%%%%%%%%%%%%%%%%%%%%%%%%%%%%%%%%%%%%%%%%%%%%%%%%%%%%%%%%%%%%%%%%%%%%%%%%%%%%%%%%%%%%%%%%%%%%%%%%%%%%%%%%%%%%%%%%%%%%%%%%%%%%%%%%%%%%%%%%%%%%%%%%%%%%%%%%%%%%%%%%%%%%%%%%%%%%%%%%%%%%%%%%%%%%%%%%%%%%%%%%%%%%%%%%%%%%%%%%%%%%%%%%%%%%%%%%%%%%%%%%%%%%%%%%%%%%%%%%%%%%%%%%%%%%%%%%%%%%%%%%%%%%%%%%%%%%%%%%%%%%%%%%%%%%%%%%%%%%%%%%%%%%%%%%%%%%%%%%%%%%%%%%%%%%%%%%%%%%%%%%%%%%%%%%%%%%%%%%%%%%%%%%%%%%%%%%%%%%%%%%%%%%%%%%%%%%%%%%%%%%%%%%%%%%%%%%%%%%%%%%%%%%%%%%%%%%%%%%%%%%%%%%%%%%%%%%%%%%%%%%%%%%%%%%%%%%%%%%%%%%%%%%%%%%%%%%%%%%%%%%%%%%%%%%%%%%%%%%%%%%%%%%%%%%%%%%%%%%%%%%%%%%%%%%%%%%%%%%%%%%%%%%%%%%%%%%%%%%%%%%%%%%%%%%%%%%%%%%%%%%%%%%%%%%%%%%%%%%%%%%%%%%%%%%%%%%%%%%%%%%%%%%%%%%%%%%%%%%%%%%%%%%%%%%%%%%%%%%%%%%%%%%%%%%%%%%
	\section{The initiation of Synchronization} \label{Sec:The initiation of Synchronization}
	
	We stratify the process of synchronization in two parts:
	\begin{itemize}
		\item {\bf the spark}, which starts the process that is going to lead to synchronization, and
		\item {\bf the buildup}, which occurs after the spark and leads to the asymptotic synchronized behavior. 
	\end{itemize}
	%In this short section, we briefly describe the build up part. The spark is discussed throughout all sections of this work.

	%%%%%%%%%%%%%%%%%%%%%%%%%%%%%%%%%%%%%%%%%%%%%%%%%%%%%%%%%%%%%%%%%%
	{\bf The spark.}
	In Sec. \ref{Sec:The model} we defined the norm of the weighted order parameter $V^t$, denoted as $r^t$, which quantifies the degree of order in the system. The value of $r^t$ varies between 0 and 1, indicating complete disorder or perfectly order, respectively, in the network dynamics. When the system is out of synchrony, $r^t$ is close to zero and varies according to the intrinsic noise inherent in the network construction, i.e., resulting from the fact that a few maps either synchronize or desynchronize purely by chance before the emergence of synchronization. After the spark of synchronization, $r^t$ might progressively increase until it reaches an asymptotic behavior. It is important to note that the sparking of synchronization is not a well-defined moment. However, in our simulations, we observe that there are sufficiently large thresholds, such that when $r^t$ exceeds that threshold, the system is highly likely to evolve toward synchronization. The "spark" happens at some point before $r^t$ reaches the threshold, and it can be defined and studied in statistical terms, as we will explain in more detail below. The "spark" is the phenomenon discussed throughout all sections of this work.
	
	{\bf The buildup.} \label{ss: build up}
	A system may remain out of synchrony for many time steps before a moment in which, spontaneously, maps start transitioning to a more synchronized dynamics. After this moment, the level of synchronization starts increasing for a few time steps before reaching the asymptotic behavior. To understand this process, we first stratify maps based on their connectivity level $w$ and compute the associated $r^t_w$ (norm of the weighted order parameter of maps with connectivity $w$, defined as $\sum_{j \in N_w}A_{ij}u^t_j=r^t_w e^{i\theta^t_w}$, where $N_w$ and $\theta_w^t$ are, respectively, the set and mean state of maps with connectivity $w$). We notice that at the beginning of the transition phase from chaos to partial synchrony, the larger contribution to $r^t$ comes from nodes of higher degree [Fig. \ref{Fig:Figure1}B]. As the system evolves, the cluster of synchrony progressively moves through less connected oscillators until reaching maps that may remain synchronized over time. This phenomenon is due to the fact that, at each time step $t$, maps with connectivity $w$ close to the ratio $2/\alpha r^t$ are more likely to synchronize (more details about this condition are presented in Appendix \ref{Sec:Appendix Sync}). If they synchronize, the level of synchronization in the system increases; i.e., $r^t$ increases. Thus, in the next time step, maps with lower connectivity $w$ will be more likely to synchronize since the fraction $2/\alpha r^t$ is reduced. As $r^t$ keeps increasing, the level $w$ at which maps are more likely to synchronize keeps decreasing. At a certain moment, the cluster of synchronization comprises maps with a certain connectivity level $w$ in which $r^t$ is close to its maximum value for that system. At that point, $r^t$ stops increasing and, maybe, even decrease, and the system is close to its asymptotic behavior. We analyzed the asymptotic properties of this phenomenon in details together with Tiago Pereira and Zheng Bian in another work \cite{Corder2023}.

	%%%%%%%%%%%%%%%%%%%%%%%%%%%%%%%%%%%%%%%%%%%%%%%%%%%%%%%%%%%%%%%%%%%%%%%%%%%%%%%%%%%%%%%%%%%%%%%%%%%%%%%%%%%%%%%%%%%%%%%%%%%%%%%%%%%%%%%%%%%%%%%%%%%%%%%%%%%%%%%%%%%%%%%%%%%%%%%%%%%%%%%%%%%%%%%%%%%%%%%%%%%%%%%%%%%%%%%%%%%%%%%%%%%%%%%%%%%%%%%%%%%%%%%%%%%%%%%%%%%%%%%%%%%%%%%%%%%%%%%%%%%%%%%%%%%%%%%%%%%%%%%%%%%%%%%%%%%%%%%%%%%%%%%%%%%%%%%%%%%%%%%%%%%%%%%%%%%%%%%%%%%%%%%%%%%%%%%%%%%%%%%%%%%%%%%%%%%%%%%%%%%%%%%%%%%%%%%%%%%%%%%%%%%%%%%%%%%%%%%%%%%%%%%%%%%%%%%%%%%%%%%%%%%%%%%%%%%%%%%%%%%%%%%%%%%%%%%%%%%%%%%%%%%%%%%%%%%%%%%%%%%%%%%%%%%%%%%%%%%%%%%%%%%%%%%%%%%%%%%%%%%%%%%%%%%%%%%%%%%%%%%%%%%%%%%%%%%%%%%%%%%%%%%%%%%%%%%%%%%%%%%%%%%%%%%%%%%%%%%%%%%%%%%%%%%%%%%%%%%%%%%%%%%%%%%%%%%%%%%%%%%%%%%%%%%%%%%%%%%%%%%%%%%%%%%%%%%%%%%%%%%%%%%%%%%%%%%%%%%%%%%%%%%%%%%%%%%%%%%%%%%%%%%%%%%%%%%%%%%%%%%%%%%%%%%%%%%%%%%%%%%%%%%%%%%%%%%%%%%%%%%%%%%%%%%%%%%%%%%%%%%%%%%%%%%%%%%%%%%%%%%%%%%%%%%%%%%%%%%%%%%%%%%%%%%%%%%%%%%%%%%%%%%%%%
	\section{The Markov process} \label{Sec:The complex-square+noise Markov chain}
	
	The main claim of this paper is that the dynamics of the weighted order parameter $V^t$ prior to partial synchronization can be closely approximated by a Markov process as in Definition \ref{Def:CSPN}.
	
	\begin{definition}
		Given parameters $k,\sigma \in \RR^+$, let 
		\begin{itemize}
			\item $v_0 = 0 \in \CC$ and 
			\item $v_{t+1} = k v_{t}^2 + \epsilon_s$,
		\end{itemize}
		where $\epsilon_s \in \CC$ follows a 2D-Gaussian distribution with mean 0 and covariance matrix $\sigma^2{\bf I}$, where ${\bf I}$ is the 2x2 identity matrix. Here, the \textit{square} operation on $v_t$ refers to the square operation on complex numbers. %We call this process as the {\em complex-square-plus-noise} dynamics.
		\label{Def:CSPN}
	\end{definition}
	
	Notice that if $||v_t||\ll1/k$, then $||v_{t+1}|| \approx \epsilon_s$ for large enough Gaussian-noise $\epsilon_s$. In this scenario, the dynamics of $v_t$ will oscillate around the origin according to the Gaussian noise. Eventually, $||v_t||$ may increase and the term $k v_{t}^2$ will have more influence on next $v_{t+1}$. If at some point, $||v_t||$ is much greater than both $1/k$ and $\sigma$, then $k ||v_{t}^2||$ will be even greater. After this moment, $v_{t}$ will very likely keep increasing forever and {\em escape} any bounded region. We are interested in studying how long it takes for this escape to happen.
	
	An important observation is that the Markov process under investigation is determined only by the product $k\sigma$. Suppose a second process governed by parameters $k'$ and $\sigma'$ such that $k'\sigma'=k \sigma$. It is not hard to see that the dynamics of the two systems is equivalent, via the transformation $v' = (k/k') v$, as the noise in the second system would have standard deviation $(k/k')\sigma = \sigma'$. For the rest of this section, we consider the case for $k=1$.
	
	Figure \ref{Fig:Figure3} illustrates that the {\em expected time to escape} for the Markov process, as in Definition \ref{Def:CSPN} with a fixed value of $\sigma$ (and $k=1$), presents a similar pattern for the expected time to escape obtained from the dynamics of the network system. We proceeded as follows. Let us consider a value $B$ large enough such that, if the $||v_t||>B$, we can be almost sure that the process will never return to a neighborhood of 0 (say $B=100$). Consider that the random variable $T$ models the time $t$ in which $||v_t||>B$. Notice that, by time $T$, the process has already escaped the neighborhoods of 0 (almost surely) forever. Notice also that the event of \textit{escaping} happens before $T$, although we do not exactly know the precise moment. However, besides some noise, we obtain from simulations that the distribution of $T$ is very close to a constant $c$ plus an geometric distribution. Importantly, the decay $\lambda$ of the geometric distribution is independent of $B$. The value of the constant $c$ is given by the number of time steps from the escaping moment itself to the moment in which $||v_t||>B$. Therefore, the constant $c$ depends on the choice of $B$. Figure \ref{Fig:Figure3} presents an example with $B=100$ (and $k=1$), which graphically implies $c=5.25$ and $\lambda=12.01$.
	
	\begin{figure}
		\centering
		\includegraphics[width=86mm]{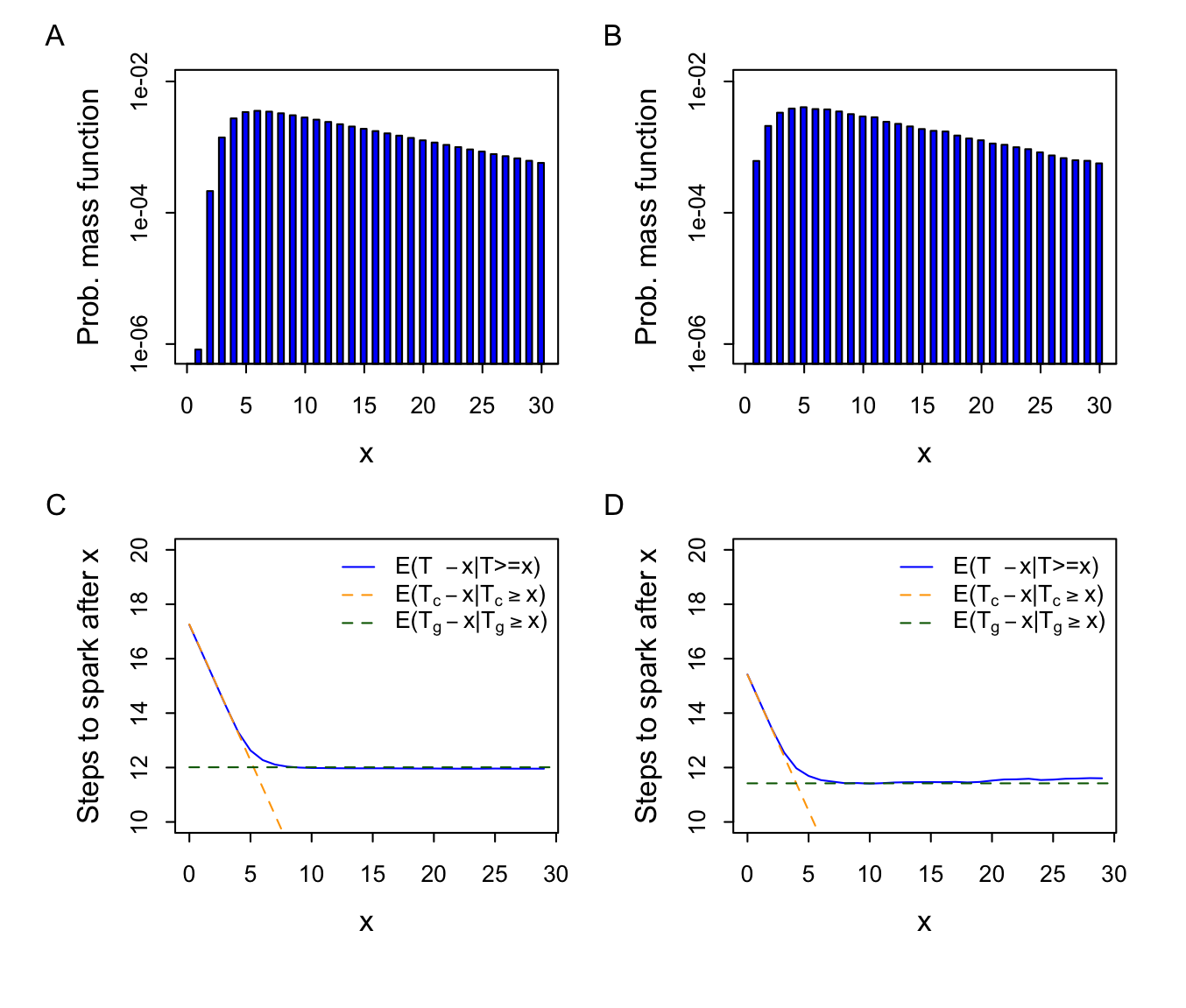}
		\caption{Expected time of escaping $\EE(\sigma)$. Panels (A) and (C) show the probability mass function and the expected value of $T-x$, respectively, given that $T>x$ plotted against $x$, where $T$ is the time to escape random-variable for the Markov process (Def. \ref{Def:CSPN}) with $\sigma = 0.36$ and $B=100$. The orange and green dashed lines represent $\EE(T_g-x|T_g>x)$ and $\EE(T_c-x|T_c>x)$, respectively, where $T_c=17.26$ and $T_g$ is a geometric random variable with mean $\lambda= 12.01$.  This shows how $T$ is very close to $5.25+T_g$. Panel (B) and (D) show the probability mass function and the same expected value of $T-x$, respectively, but with $T$ representing the time-to-escape obtained from a simulation of a network with $N = 50,000$, $C = 500$, $\gamma = 4$, and $\alpha = 12$ and dynamics governed by Eq. \ref{Eq:M0}. Orange and green dashed lines consider values $T_c=15.41$ and $\lambda= 11.42$.  
			\label{Fig:Figure3}
		}
	\end{figure}
	
	To show that the distribution of the random variable $T$ is  close to a constant plus a geometric distribution, in Fig. \ref{Fig:Figure3} we illustrate the expected value $\EE(T-x|T>x)$ of the remaining number of time steps, given that the number of steps is above $x$. One can see that from some fixed value of $x$ onward, $\EE(T-x|T>x)$ is constant, and, thus, $T$ is memoryless like the geometric distribution. Recall that the geometric is the only distribution on $\mathbb{N}$ for which $\EE(T-x|T>x)$ is independent of $x$. Notice that for small values of $x$, $\EE(T-x|T>x)$ is linear in $x$ with a slope of $-1$. 
	
	\begin{definition}
		We define the {\em expected time to escape} as the expected value of the geometric part of the distribution of a random variable $T$, i.e., the value of $\EE(T-x| T> x)$ for some large enough $x$. We call this function $\EE(\sigma)$.
	\end{definition}
	
	Importantly, we do not derive a closed form equation for the function $\EE(\sigma)$.

	%%%%%%%%%%%%%%%%%%%%%%%%%%%%%%%%%%%%%%%%%%%%%%%%%%%%%%%%%%%%%%%%%%%%%%%%%%%%%%%%%%%%%%%%%%%%%%%%%%%%%%%%%%%%%%%%%%%%%%%%%%%%%%%%%%%%%%%%%%%%%%%%%%%%%%%%%%%%%%%%%%%%%%%%%%%%%%%%%%%%%%%%%%%%%%%%%%%%%%%%%%%%%%%%%%%%%%%%%%%%%%%%%%%%%%%%%%%%%%%%%%%%%%%%%%%%%%%%%%%%%%%%%%%%%%%%%%%%%%%%%%%%%%%%%%%%%%%%%%%%%%%%%%%%%%%%%%%%%%%%%%%%%%%%%%%%%%%%%%%%%%%%%%%%%%%%%%%%%%%%%%%%%%%%%%%%%%%%%%%%%%%%%%%%%%%%%%%%%%%%%%%%%%%%%%%%%%%%%%%%%%%%%%%%%%%%%%%%%%%%%%%%%%%%%%%%%%%%%%%%%%%%%%%%%%%%%%%%%%%%%%%%%%%%%%%%%%%%%%%%%%%%%%%%%%%%%%%%%%%%%%%%%%%%%%%%%%%%%%%%%%%%%%%%%%%%%%%%%%%%%%%%%%%%%%%%%%%%%%%%%%%%%%%%%%%%%%%%%%%%%%%%%%%%%%%%%%%%%%%%%%%%%%%%%%%%%%%%%%%%%%%%%%%%%%%%%%%%%%%%%%%%%%%%%%%%%%%%%%%%%%%%%%%%%%%%%%%%%%%%%%%%%%%%%%%%%%%%%%%%%%%%%%%%%%%%%%%%%%%%%%%%%%%%%%%%%%%%%%%%%%%%%%%%%%%%%%%%%%%%%%%%%%%%%%%%%%%%%%%%%%%%%%%%%%%%%%%%%%%%%%%%%%%%%%%%%%%%%%%%%%%%%%%%%%%%%%%%%%%%%%%%%%%%%%%%%%%%%%%%%%%%%%%%%%%%%%%%%%%%%%%%%%%%%%
	\section{The main result} \label{Sec:The main result}
	
	Here, we write the Markov process presented in Sec. \ref{Sec:The complex-square+noise Markov chain} now accounting for the network structure and its dynamics. Given a network design defined by parameters $N$, $C$, $\alpha$, and $\delta$, the dynamics of the weighted order parameter $V^t$ before synchronization behaves as the Markov process as in Definition \ref{Def:CSPN} with the following parameters:
	\begin{equation} \label{Def:CSPN_parameters}
		\begin{split}
			\sigma & = \sqrt{M_2(\delta)/(2N)} \text{ and} \\     
			k & = \alpha^2\  K^\alpha_\delta(C) \ M_3(\delta) (1+ 4/\alpha M_2(\delta))/8,
		\end{split}
	\end{equation}
	where $M_2(\delta) = \int_0^{\infty} w^2 d\delta(w)$ and $M_3(\delta)=\int_0^{\infty} w^3 d\delta(w)$ are, as defined before, the second and third cumulative moments of the probability density function $\delta$, respectively. $K^\alpha_\delta(C)$ (derived in Sec. \ref{Sec:The continuous approximation}) is the only parameter that depends on $C$, and it measures the impact of the noise resulting from the fact that each map is connected to a distinct set of neighbors. 
	
	Figure \ref{Fig:Figure2}A presents computational simulations performed for different network designs with degrees randomly sampled from inverse-gamma distributions and shows that the expected time to synchronization can be approximated by the Markov process as in Definition \ref{Def:CSPN}, with parameters defined based on the network structure as presented in \ref{Def:CSPN_parameters}.
	
	As we will see in Sec. \ref{Sec:The continuous approximation}, the expected time to spark of synchronization can, thus, be approximated by $\EE(k\sigma)$, which is approximated by,
	\begin{equation}
		\EE\left(  \frac{\alpha^2}{\sqrt{2N}}\  M^\alpha(\delta) \  K^\alpha_\delta(C)/8 \right),
		\label{Eq:E}
	\end{equation}
	where $M^\alpha(\delta)$ [Eq. \ref{Eq:M}; derived in Sec. \ref{Sec:The continuous approximation}] is a term that measures the effect of the sparsity of the network degree distribution $\delta$ on the expected time to spark. Figure \ref{Fig:Figure4} presents how $M^\alpha(\delta)$ and $K^\alpha_\delta(C)$ [Eq. \ref{Eq:K}] vary according to $\delta$ and $C$.
	
	\begin{figure}[h]
		\centering
		\includegraphics[width=86mm]{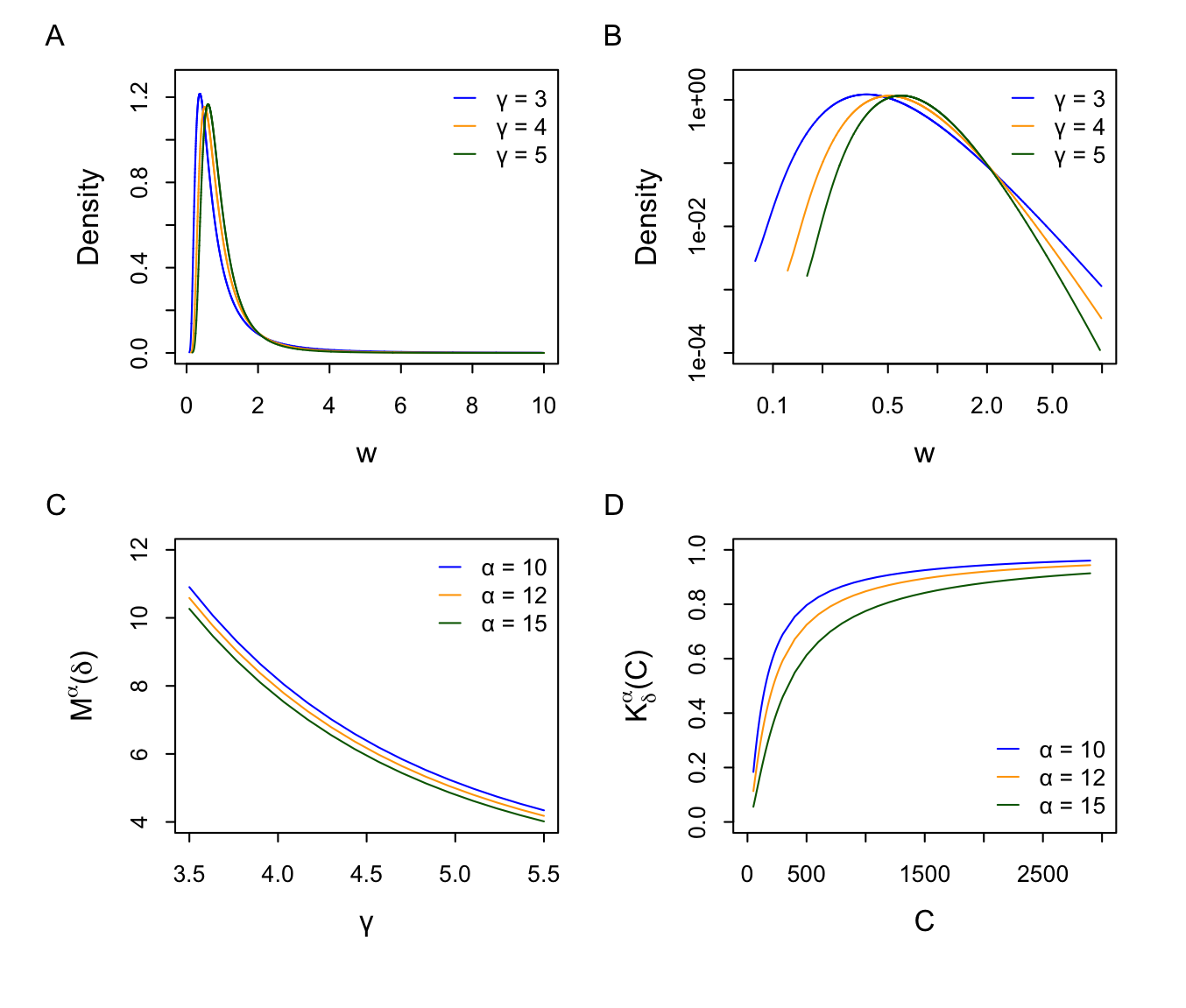}
		\caption{The parameters of the Markov process under investigation. Panels (A) and (B) show the shape of the inverse gamma distributions in linear--linear and log--log scales, respectively, for power-law exponents $\gamma=3$ (blue curve), $\gamma=4$ (yellow), and $\gamma=5$ (green). Panel (C) illustrates how $M^\alpha(\delta)$ [Eq. \ref{Eq:M}] varies according to the power-law exponent $\gamma$. Panel (D) illustrates how $K^\alpha_\delta(C)$ [Eq. \ref{Eq:K}] varies according to the network mean degree $C$. In panels (C) and (D), blue, orange, and green curves correspond to $\alpha = 10, 12$, and $15$, respectively.
			\label{Fig:Figure4}
		}
	\end{figure}
	
	Eq. \ref{Eq:E} shows how the values $\alpha$, $C$, $N$, and the probability distribution $\delta$ are related to the expected time to spark. Some observations are as follows: (i) the larger the value of the coupling strength $\alpha$, the faster the spark; (ii) the larger the number of maps $N$ in a network, the slower the spark; (iii) the time to spark is correlated with the square root of the number of maps $N$ and to the square of the network coupling strength $\alpha$; (iv) the more dispersed the network degree distribution [larger values of $M_3(\delta)$], the faster the spark; and (v) the term $\sqrt{M_2(\delta)} + \frac{4}{\alpha\sqrt{M_2(\delta)}}$ is smallest when $M_2(\delta)=4/\alpha$. Since $M_2(\delta)>1>4/\alpha$, this term also increases with $M_2(\delta)$ (dispersion of $\delta$).
	
	Notice that the parameter $\sigma$ is the standard deviation of the weighted order parameter $V^t$, calculated for $N$ maps with states randomly sampled from a uniform distribution, and on a network with degree distribution $\delta$. Therefore, $\sigma = \sqrt{M_2(\delta)/(2N)}$. We argue that this is a good approximation since the state of maps before synchrony can be approximated by a uniform distribution. In Sec. \ref{Sec:The continuous approximation} we deeply investigate the dynamics of the system and present the derivation of the formula for $k$.

	%%%%%%%%%%%%%%%%%%%%%%%%%%%%%%%%%%%%%%%%%%%%%%%%%%%%%%%%%%%%%%%%%%%%%%%%%%%%%%%%%%%%%%%%%%%%%%%%%%%%%%%%%%%%%%%%%%%%%%%%%%%%%%%%%%%%%%%%%%%%%%%%%%%%%%%%%%%%%%%%%%%%%%%%%%%%%%%%%%%%%%%%%%%%%%%%%%%%%%%%%%%%%%%%%%%%%%%%%%%%%%%%%%%%%%%%%%%%%%%%%%%%%%%%%%%%%%%%%%%%%%%%%%%%%%%%%%%%%%%%%%%%%%%%%%%%%%%%%%%%%%%%%%%%%%%%%%%%%%%%%%%%%%%%%%%%%%%%%%%%%%%%%%%%%%%%%%%%%%%%%%%%%%%%%%%%%%%%%%%%%%%%%%%%%%%%%%%%%%%%%%%%%%%%%%%%%%%%%%%%%%%%%%%%%%%%%%%%%%%%%%%%%%%%%%%%%%%%%%%%%%%%%%%%%%%%%%%%%%%%%%%%%%%%%%%%%%%%%%%%%%%%%%%%%%%%%%%%%%%%%%%%%%%%%%%%%%%%%%%%%%%%%%%%%%%%%%%%%%%%%%%%%%%%%%%%%%%%%%%%%%%%%%%%%%%%%%%%%%%%%%%%%%%%%%%%%%%%%%%%%%%%%%%%%%%%%%%%%%%%%%%%%%%%%%%%%%%%%%%%%%%%%%%%%%%%%%%%%%%%%%%%%%%%%%%%%%%%%%%%%%%%%%%%%%%%%%%%%%%%%%%%%%%%%%%%%%%%%%%%%%%%%%%%%%%%%%%%%%%%%%%%%%%%%%%%%%%%%%%%%%%%%%%%%%%%%%%%%%%%%%%%%%%%%%%%%%%%%%%%%%%%%%%%%%%%%%%%%%%%%%%%%%%%%%%%%%%%%%%%%%%%%%%%%%%%%%%%%%%%%%%%%%%%%%%%%%%%%%%%%%%%%%%%%%
	\section{The continuous approximation} \label{Sec:The continuous approximation}
	
	In this section, we dive into a deep investigation of the network dynamics of coupled maps governed by Eq. \ref{Eq:M1} and elucidate the motivations behind the claims made in Secs. \ref{Sec:The model}--\ref{Sec:The main result}. We will derive a continuous approximation to Eq. \ref{Eq:M1} such that the states of maps are well approximated by a probability distribution that evolves over time. We notice, however, that we are not deriving the limit as the size $N$ of the graphs converges to $\infty$. Instead, we assume that $N$ is sufficiently large such that the continuous version of the model provides a good approximation of its dynamics, but not so large as to eliminate the stochastic behavior inherent in a finite-sized network. We will demonstrate that small random variances in the degree distribution introduced by random networks play a crucial role in the sparking of synchronization.

	%%%%%%%%%%%%%%%%%%%%%%%%%%%%%%%%%%%%%%%%%%%%%%%%%%%%%%%%%%%%%%%%%%
	%\subsection{Acting on continuous probability distributions}
	
	{\bf Acting on continuous probability distributions.}
	For each connectivity layer $w\in\RR^+$, we consider a probability density function $\rho_w(z)$ with $z\in S^1$. We use $\rho$ to denote the probability density for the pair $(w,z)\in \RR^+\times S^1$, denoted as $\rho(w,z)= \rho_w(z)\delta(w)$. This function evolves on time (we omit the superscript $t$) as the dynamical system acts on the space of probability density functions $\rho$ on $ \RR^+\times S^1$ that are compatible with the network degree distribution $\delta$.
	
	Consider a fixed value of $w$ and denote $\tau(z) = \rho_w(z)$ as a probability density function on $S^1$. Given the order parameter $V \in \CC$, we define an operator $\FF_{V}$ acting on probability distribution functions defined on $S^1$. $\FF_{V}(\tau)$ describes the probability distribution on $S^1$ obtained by applying one time step of the dynamics. To this end, for each map $i$ with connectivity $w$ and state $z_i\in S^1$, we apply the function
	\begin{equation}
		z_i\mapsto 2z_i + \frac{\alpha}{C}\Im(V_i \bar{u}_i),
		\label{Eq:zi}
	\end{equation}
	where $\bar{u}_i = e^{-iz_i}$ and $V_i$ is the sum of all states $u_j$s, with $j$ connected to $i$ (as described in Sec. \ref{Sec:The model}). Recall from Sec. \ref{Sec:The model} that $V_i$ is distributed in the complex plane according to a 2D Gaussian distribution with mean $d_iV$ and covariance matrix $(d_i/2){\bf I}$, where $d_i= wC$ is the number of connections of map $i$. It implies that $\Im(V_i \bar{u}_i)$ is distributed following a 1D Gaussian distribution with mean $wC\Im(V \bar{u}_i)$ and variance $wC/2$. Thus, the dynamics can be described by mapping $z_i \mapsto 2z_i + \alpha w\ \Im(V \bar{u}_i) + \epsilon_i$, where $\epsilon_i$ is a 1D Gaussian random variable with mean $0$ and variance $\alpha^2w/2C$.
	
	We decompose the operator $\FF_{V}$ in two operators also acting on probability distribution functions defined on $S^1$ as follows:
	\begin{equation}
		\FF_{V} = \DD_{\sqrt{\alpha^2w/2C}}\circ \MM_{\alpha wV}.
		\label{Eq:Fv}
	\end{equation}
	The main operator $\MM_P$, for $P\in \CC$, maps a distribution $\tau$ to the distribution obtained after applying the function $z_i\mapsto 2z_i + \Im(P \bar{u}_i)$. The diffusion operator $\DD_{\nu}$ maps a distribution $\tau$ to the distribution we get from applying the function $z_i\mapsto z_i + \epsilon_i$, where $\epsilon_i$ is a normal random variable with mean $0$ and standard deviation $\nu$.  In this framework, the mean field $\EE_{\tau}$ at each layer $w$ can be obtained, at each time step $t$, as in Def. \ref{Def:Mean_field}.
	
	\begin{definition}
		Given a probability distribution function $\tau$ on $S^1$, let $\EE_{\tau}$ be the expected value of $e^{iz}$ where $z \sim \tau$, that is,
		\[
		\EE_{\tau} = \int_{S^1} e^{iz}\tau(z)dz.
		\]
		\label{Def:Mean_field}
	\end{definition}
	
	Next, we analyze the operators $\DD_\nu$ and $\MM_P$.

	%%%%%%%%%%%%%%%%%%%%%%%%%%%%%%%%%%%%%%%%%%%%%%%%%%%%%%%%%%%%%%%%%%
	%\subsection{The diffusion operator $\DD_\nu$}
	
	{\bf The diffusion operator $\DD_\nu$.}
	$\DD_{\nu}(\tau)$ is the distribution obtained after mapping a distribution $\tau$ to the function $z_i\mapsto z_i + \epsilon_i$, where $\epsilon_i$ is a 1D Gaussian random variable with mean $0$ and standard deviation $\nu$. The diffusion operator acts on the mean field $\EE_{\tau}$ by shrinking its length by a factor of $e^{-\nu^2/2}$ but keeping its direction (Lemma \ref{L:Dp}).
	
	\begin{lemma}
		$\EE_{\DD_{\nu}(\tau)} = e^{-\nu^2/2} \EE_{\tau}$ {\normalfont(proof in Appendix \ref{Sec:Appendix Proofs}).}
		\label{L:Dp}
	\end{lemma}

	%%%%%%%%%%%%%%%%%%%%%%%%%%%%%%%%%%%%%%%%%%%%%%%%%%%%%%%%%%%%%%%%%%
	%\subsection{The main operator $\MM_P$}
	
	{\bf The main operator $\MM_P$.}
	We study the main operator $\MM_P$ on probability distributions that are very close to the uniform distribution, i.e., the case before the spark of synchronization. Since $z\in [0,2\pi]$, the uniform distribution is given by $\tau_0(z)= 1/2\pi$.
	\begin{definition}
		Given $a = \epsilon e^{i\theta}\in \CC$, let $\tau_a$ be the probability distribution on $S^1$ of the form,
		\[
		\tau_a(z) = \frac{1}{2\pi}\big( 1 + 2\epsilon \cos(z - \theta)\big).
		\]
		\label{Def:tau}
	\end{definition}
	Notice that we can also write $\tau_a(z) = \frac{1}{2\pi}\big( 1 + 2\Re(e^{iz} \bar{a})\big)$, where $\bar{a}$ is the complex conjugate of $a$, and thus, the uniform distribution corresponds to $\tau_a$ for $a=0$. Interestingly, we will see below that the application of $\MM_P$ or $\DD_\nu$ on a probability distribution of the form $\tau_a(z)$ maps it to another distribution of the form $\tau_{a'}(z)$ up to a small error. Before, we show on Lemma \ref{L:E} that the expected value of $\tau_a$ is $a$.
	
	\begin{lemma} 
		$\EE_{\tau_a}=a$ {\normalfont(proof in Appendix \ref{Sec:Appendix Proofs})}.
		\label{L:E}
	\end{lemma}
	
	The distribution $\tau_a$ is essentially the simplest distribution on $S^1$ that one can think of with mean field $a$. Thus, we assume that, before synchronization, there exists a complex number $a_w$ for each $w$, with norm of $a_w$ small, such that $\rho_w = \tau_{a_w}$. Lemmas \ref{L:Mp} and \ref{L:Dsigma} present how the operators $\MM_P$ and $\DD_\nu$, respectively, act on distributions of the form $\tau_a$.
	
	\begin{lemma}
		Given $a \in \CC$ and $P\in \CC$ ,
		\[
		\MM_P(\tau_a) \approx \tau_{a'},
		\]
		where $a' = P(P+4a)/8$ {\normalfont(proof in Appendix \ref{Sec:Appendix Proofs})}.
		\label{L:Mp}
	\end{lemma}
	
	\begin{lemma}
		Given $a \in \CC$ and $\nu\in \RR^+$,
		\[
		\DD_\nu(\tau_a) \approx \tau_{a'},
		\]
		where $a' = e^{-\nu^2/2} a$ {\normalfont(proof in Appendix \ref{Sec:Appendix Proofs})}.
		\label{L:Dsigma}
	\end{lemma}

	%%%%%%%%%%%%%%%%%%%%%%%%%%%%%%%%%%%%%%%%%%%%%%%%%%%%%%%%%%%%%%%%%%
	%\subsection{The dynamics}
	
	{\bf The dynamics.}
	Here, we consider all connectivity layers $w$ and compute the dynamics of the order parameter $V$. We can calculate the order parameter at the next time step, $V^{+}$, and show that the map $V \mapsto V^{+}$ can be well approximated by the Markov process defined in Sec. \ref{Sec:The complex-square+noise Markov chain}. Since we are dealing with the dynamics of the network before the sparking of synchronization, we assume that all distributions $\rho_w$ are close to uniform. For each $w$, let $a_w = \EE_{\rho_w}$ (close to $0$ before synchronization). Notice that $a_w$  must satisfy
	\begin{equation}
		\int w a_w \delta(w) dw = V.
		\label{Eq:int_aw}
	\end{equation}
	We also assume that $\rho_w$ is of the form $\tau_{a_w}$, for all $w$, and calculate $\rho^{+}_w$ (density $\rho_w$ at the next time step) as follows:
	\[
	\rho^{+}_w = \FF_{V}(\rho_w) = \DD_{\alpha\sqrt{w/2C}}(\MM_{\alpha w V}(\tau_{a_w})) =\tau_{a'_w}
	\]
	where $a'_w = e^{-\alpha^2 w/4C} \alpha wV (\alpha wV+4a_w)/8$. Therefore, the expected value for $V^{+}$ is obtained by
	\begin{eqnarray}\label{eq: V+ formula}
		V^{+} &=& \int w \left(e^{-\alpha^2 w/4C} \alpha wV (\alpha wV+4a_w)/8 \right)\delta(w) dw     \nonumber       \\ 
		&=& \alpha^2 (V)^2 \int \left(w^3 e^{-\alpha^2 w/4C} \delta(w)/8\right) dw \nonumber \\
		&&+ \alpha V \int \left(w^2 a_w e^{-\alpha^2 w/4C} \delta(w) /2\right) dw. 
		\label{E:Vplus}
	\end{eqnarray}
	
	Based on Eq. \ref{E:Vplus}, we define $K^\alpha_\delta(C)$ as in Eq. \ref{Eq:K}.
	%\begin{equation}
	%   K^\alpha_\delta(C) = \frac{\int_0^{\infty} w^3 e^{-w \frac{\alpha^2}{4C}} \delta(w) dw}{\int_0^{\infty} w^3 \delta(w) dw}.
	%    \label{Eq:int_aw}
	%\end{equation}
	The term $K^\alpha_\delta(C)$ is relevant only for small values of $C$. If the mean degree $C$ is large, $e^{-w \frac{\alpha^2}{4C}} \approx 1$, and thus, $K^\alpha_\delta(C) \approx 1$. For small values of $C$, $K^\alpha_\delta(C)$ is always less than $1$; i.e., it slows down the sparking time. The reason for this lies in the fact that each map $i$ receives signals from a distinct set of $w_iC$ neighbors, resulting in some random noise in the received signals for each map. This noise has a detrimental effect on the synchronization of the system.
	
	The second term of Eq. \ref{E:Vplus} has the random variable $a_w$. We show in Appendix \ref{Sec:Appendix aw} that $\int w^2 a_w e^{-\alpha^2 w/4C} \delta(w)$ has expected value $V M_3(\delta) K^\alpha_\delta(C)/M_2(\delta)$.

	%%%%%%%%%%%%%%%%%%%%%%%%%%%%%%%%%%%%%%%%%%%%%%%%%%%%%%%%%%%%%%%%%%
	%\subsection{The variance on the order parameter}
	
	{\bf The variance on the order parameter.}
	After computing the distribution $\rho^{+}$, we take $N$ random points $(w_i,z_i)$. The next order parameter is then obtained as $V^{+} = \sum_{i=1}^N w_i e^{iz_i}$. The expected value of $w_i e^{iz_i}$ is $V'=(\alpha^2 M_3(\delta) K^\alpha_\delta(C) V^2/8)(1+ 4/\alpha M_2)$. As for the variance, we obtain $\EE(w_i^2)\VV(e^{iz_i})$. Since $\rho^{+}$ is very close to the uniform distribution, we can approximate the variance $\VV_{\rho^{+}}(e^{iz_i}) \approx \VV_{unif}(e^{iz_i}) = (1/2){\bf I}$, where ${\bf I}$ is the 2x2 identity matrix. Moreover,
	\begin{equation}
		\EE(w_i^2) = \int w^2\delta(w)dw = M_2(\delta)
		\label{Eq:E_wi2}
	\end{equation}
	is the second momentum of the degree distribution $\delta$. Taking the mean of $N$ points, we can use the central-limit theorem and conclude that $V^{+}$ is 2D-Gaussian distributed in the complex plane with mean $V'$ and variance $M_2(\delta)/2N$. 
	
	Therefore, by taking $\sigma = \sqrt{M_2(\delta)/(2N)}$ and $k = \alpha^2\  K^\alpha_\delta(C) \ M_3(\delta) (1+ 4/\alpha M_2(\delta))/8$ (as presented in \ref{Def:CSPN_parameters}), the expected time to the spark of synchronization can, thus, be approximated by Eq. \ref{Eq:E}, with $M^{\alpha}(\delta)$ and $K^\alpha_\delta(C)$ as defined in Eqs. \ref{Eq:M} and \ref{Eq:K}, respectively.

	%%%%%%%%%%%%%%%%%%%%%%%%%%%%%%%%%%%%%%%%%%%%%%%%%%%%%%%%%%%%%%%%%%%%%%%%%%%%%%%%%%%%%%%%%%%%%%%%%%%%%%%%%%%%%%%%%%%%%%%%%%%%%%%%%%%%%%%%%%%%%%%%%%%%%%%%%%%%%%%%%%%%%%%%%%%%%%%%%%%%%%%%%%%%%%%%%%%%%%%%%%%%%%%%%%%%%%%%%%%%%%%%%%%%%%%%%%%%%%%%%%%%%%%%%%%%%%%%%%%%%%%%%%%%%%%%%%%%%%%%%%%%%%%%%%%%%%%%%%%%%%%%%%%%%%%%%%%%%%%%%%%%%%%%%%%%%%%%%%%%%%%%%%%%%%%%%%%%%%%%%%%%%%%%%%%%%%%%%%%%%%%%%%%%%%%%%%%%%%%%%%%%%%%%%%%%%%%%%%%%%%%%%%%%%%%%%%%%%%%%%%%%%%%%%%%%%%%%%%%%%%%%%%%%%%%%%%%%%%%%%%%%%%%%%%%%%%%%%%%%%%%%%%%%%%%%%%%%%%%%%%%%%%%%%%%%%%%%%%%%%%%%%%%%%%%%%%%%%%%%%%%%%%%%%%%%%%%%%%%%%%%%%%%%%%%%%%%%%%%%%%%%%%%%%%%%%%%%%%%%%%%%%%%%%%%%%%%%%%%%%%%%%%%%%%%%%%%%%%%%%%%%%%%%%%%%%%%%%%%%%%%%%%%%%%%%%%%%%%%%%%%%%%%%%%%%%%%%%%%%%%%%%%%%%%%%%%%%%%%%%%%%%%%%%%%%%%%%%%%%%%%%%%%%%%%%%%%%%%%%%%%%%%%%%%%%%%%%%%%%%%%%%%%%%%%%%%%%%%%%%%%%%%%%%%%%%%%%%%%%%%%%%%%%%%%%%%%%%%%%%%%%%%%%%%%%%%%%%%%%%%%%%%%%%%%%%%%%%%%%%%%%%%%%%%
	\section{Discussion} \label{Sec:Discussion}
	We studied the mechanisms for the spark of synchronization in heterogeneous networks of Bernoulli coupled maps. We obtained a continuum limit approximation of the model and derived an expression whose expected value determines the sparking time of synchrony given a network design. We showed with numerical simulations in finite-sized networks that our results are consistent with the analysis. %Interestingly, all our results are directly applied to heterogeneous networks with heterogeneous coupling strength.
	
	The correlation of the expected sparking time of synchrony and the network construction [Eq. \ref{Eq:E}] shows that, for a fixed coupling strength, mean degree and degree distribution, the larger the network, the longer is the expected time the system takes to synchronize. Furthermore, the time it takes to spark is more influenced by changes in the coupling strength (powered-two correlation) rather than the number of maps in the network (squared-root correlation). One can also see that the more heterogeneous the network degree distribution [i.e., the larger $M_2(\delta)$ and $M_3(\delta)$], the sooner the system is expected to synchronize. In addition, larger mean degrees lead to faster synchronization since it implies larger values of $K^\alpha_\delta(C)$; however, the increase in $K^\alpha_\delta(C)$ is rather slow for larger values of $C$, and its variation after a certain threshold is expected to have small effects on the time to synchronization.
	
	Our results are consistent with others \cite{Pereira2017, Restrepo2006} and show that the emergence of coherence in large collections of heterogeneous coupled chaotic systems is also highly dependent on the network topology. However, while most analysis of the synchronization in coupled networks with interacting dynamical systems investigates the critical coupling strength, we looked at the problem from a new perspective. By studying the dynamics of systems in heterogeneous networks before synchrony, we derived an expression that correlates the parameters of the network structure with the expected time to synchronization. It is important to note that due to the intrinsic error inherent in finite-sized networks being modeled with normal distributions, there is always a chance that a network may lose synchronization. However, the network dynamics may spontaneously return to synchrony quickly for some identifiable configurations. This analysis could be pertinent for various applications where synchrony among a set of individuals is a desirable condition, particularly those where one can control the network design.
	
	To our knowledge, this is the first attempt to investigate the expected time for a network to reach synchrony considering heterogeneous networks without a fixed form for the degree distribution. Our simulations show that the Markov process presented in this work is a suitable tool to model the sparking of synchronization in heterogeneous networks of coupled maps. Our results, however, present a small deviation [approximately 10\%; Fig. \ref{Fig:Figure2}A] which does not affect qualitatively the conclusions derived from the analysis. Future investigations may try to understand the reasons for this deviation as well as derive a closed-form equation for the expected time to synchronization. Whether to study the dynamics of neurons \cite{Iaccarino2016, Martorell2019}, power grids \cite{Motter2013} or any other heterogeneous networks of coupled units, we argue that understanding the fundamental mechanisms for the spark of synchrony can shed light into new control strategies that enhance coherence between interacting units.

	\section{Acknowledgments}
	We thank Tiago Pereira for enlightening discussions and supporting to present our results.

	%%%%%%%%%%%%%%%%%%%%%%%%%%%%%%%%%%%%%%%%%%%%%%%%%%%%%%%%%%%%%%%%%%%%%%%%%%%%%%%%%%%%%%%%%%%%%%%%%%%%%%%%%%%%%%%%%%%%%%%%%%%%%%%%%%%%%%%%%%%%%%%%%%%%%%%%%%%%%%%%%%%%%%%%%%%%%%%%%%%%%%%%%%%%%%%%%%%%%%%%%%%%%%%%%%%%%%%%%%%%%%%%%%%%%%%%%%%%%%%%%%%%%%%%%%%%%%%%%%%%%%%%%%%%%%%%%%%%%%%%%%%%%%%%%%%%%%%%%%%%%%%%%%%%%%%%%%%%%%%%%%%%%%%%%%%%%%%%%%%%%%%%%%%%%%%%%%%%%%%%%%%%%%%%%%%%%%%%%%%%%%%%%%%%%%%%%%%%%%%%%%%%%%%%%%%%%%%%%%%%%%%%%%%%%%%%%%%%%%%%%%%%%%%%%%%%%%%%%%%%%%%%%%%%%%%%%%%%%%%%%%%%%%%%%%%%%%%%%%%%%%%%%%%%%%%%%%%%%%%%%%%%%%%%%%%%%%%%%%%%%%%%%%%%%%%%%%%%%%%%%%%%%%%%%%%%%%%%%%%%%%%%%%%%%%%%%%%%%%%%%%%%%%%%%%%%%%%%%%%%%%%%%%%%%%%%%%%%%%%%%%%%%%%%%%%%%%%%%%%%%%%%%%%%%%%%%%%%%%%%%%%%%%%%%%%%%%%%%%%%%%%%%%%%%%%%%%%%%%%%%%%%%%%%%%%%%%%%%%%%%%%%%%%%%%%%%%%%%%%%%%%%%%%%%%%%%%%%%%%%%%%%%%%%%%%%%%%%%%%%%%%%%%%%%%%%%%%%%%%%%%%%%%%%%%%%%%%%%%%%%%%%%%%%%%%%%%%%%%%%%%%%%%%%%%%%%%%%%%%%%%%%%%%%%%%%%%%%%%%%%%%%%%%%%%
	\appendix
	\section{Simulation of the expected time to the spark of synchronization} \label{Sec:Appendix Simulation}
	We present the algorithm we used for obtaining the expected time to synchronization (Alg. \ref{Alg:ES}). Its input consists of the adjacency matrix $A_{ij}$, the network coupling strength $\alpha$, the states initial condition $z^0_i$ (randomly sampled from a uniform distribution ${\mathcal U}_{[0,2\pi]}$), the maximum number of iterations $t_{max}$, and the parameter $r_c$ that determines when the spark of synchronization occurs. Because $r^t$ has different asymptotic behaviors for different network configurations, to determine the value $r_c$ for each network configuration, we first calculated the average value of $r$ after synchronization and defined $r_c$ as half of it. We use half since the value of $r$ at synchronization is not completely stable. For the simulations presented in Figs. \ref{Fig:Figure2} and \ref{Fig:Figure3}, we considered $t_{max}=100,000$.

	\begin{algorithm}
		\caption{Time to emergence of synchrony}\label{Alg:ES}
		\begin{algorithmic}[1] 
			\Require $z^0_i, \alpha, A_{ij}, r_c, t_{max}$
			\Ensure $t_s$
			\State $count \gets 0$
			\For{$t \in 0,...,t_{max}$}       
			\State $\mbox{Compute } r^t$ 
			\If{$r^t \geq r_c$} 
			\State $count \gets count+1$
			\State $z^t \gets N \mbox{ randomly sampled values from } {\mathcal U}_{[0,2 \pi]}$
			\Else
			\State $\mbox{Update } z^t$ 
			\EndIf
			\EndFor
			\State $t_s \gets t_{max}/count$
		\end{algorithmic}
	\end{algorithm}

	%%%%%%%%%%%%%%%%%%%%%%%%%%%%%%%%%%%%%%%%%%%%%%%%%%%%%%%%%%%%%%%%%%
	\section{Maps with connectivity $w$ close to the ratio $2/\alpha r^t$ are more likely to synchronize} \label{Sec:Appendix Sync}
	Recall from Sec. \ref{Sec:The model} that $V^t_i \sim {\mathcal N}^2(d_i V^t, (d_i/2) {\bf I})$. Thus, we can approximate Eq. \ref{Eq:M1} as follows:
	\begin{equation}
		z^{t+1}_i \approx 2 z^t_i + \alpha w_i \Im(V^t \bar{u}^t_i),
		\label{Eq:M2}
	\end{equation}
	with $V^t=r^t e^{i\theta^t}$. For all maps $i$, we can write $z^t_i = \theta^t + h^t_i$, with $h^t_i$ baing small for those maps that are synchronized at time $t$. Thus, by approximating Eq. \ref{Eq:M2} by its first degree [i.e., using $sin(h) \approx h$ for $h \approx 0$], we obtain $z^{t+1}_i \approx 2\theta^t + h_i(2-w_i\alpha r^t)$. The closest the term $2-w_i\alpha r^t$ is to $0$, the more likely a map is to synchronize, i.e., $z^{t+1}_i \approx 2\theta^t$. Therefore, maps with connectivity $w$ close to the ratio $2/\alpha r^t$ are more likely to synchronize.

	%%%%%%%%%%%%%%%%%%%%%%%%%%%%%%%%%%%%%%%%%%%%%%%%%%%%%%%%%%%%%%%%%%
	\section{Proofs of the Lemmas} \label{Sec:Appendix Proofs}
	%{\bf Proofs of the Lemmas.}
	Recalling Definitions \ref{Def:Mean_field} and \ref{Def:tau}, here we present the proofs of Lemmas \ref{L:Dp}, \ref{L:E}, \ref{L:Mp}, and \ref{L:Dsigma}. 
	
	\begin{lemma*} 
		$\EE_{\tau_a}=a$.
	\end{lemma*}
	\begin{proof}%[Proof of Lemma \ref{L:E}.]
		%Write $a$ as $\epsilon e^{i\theta}$.
		\begin{eqnarray*}
			\EE_{\tau_a} &=& \int_{-\pi}^\pi e^{iz} \tau_a(z) dz \\
			&=& \frac{1}{2\pi}\int_{-\pi}^\pi e^{iz} (1+2\epsilon \cos(z - \theta))dz             \\
			&=&  \frac{1}{2\pi}\int_{-\pi}^\pi e^{iz} dz + \frac{\epsilon}{\pi}\int_{-\pi}^\pi e^{i(z'+\theta)} \cos(z') dz'     \\
			&=& \frac{\epsilon}{\pi}e^{i\theta} \int_{-\pi}^\pi e^{i z'}\cos(z') dz'   \\
			&=& \frac{\epsilon}{\pi}e^{i\theta}\left(\int_{-\pi}^\pi \cos(z')^2 dz' + i\int_{-\pi}^\pi \sin(z') \cos(z') dz' \right)   \\
			&=&   \frac{\epsilon}{\pi}e^{i\theta}\big( \pi + 0\big)    \\
			&=&   \epsilon e^{i\theta}  = a                         \qedhere
		\end{eqnarray*}
	\end{proof}
	
	\begin{lemma*}
		Given $a \in \CC$ and $\nu\in \RR^+$,
		\[
		\DD_\nu(\tau_a) \approx \tau_{a'},
		\]
		where $a' = e^{-\nu^2/2} a$.
	\end{lemma*}
	\begin{proof}%[Proof of Lemma \ref{L:Dp}.]
		Let $G_\nu(x)$ be the Gaussian density function $\frac{1}{\nu\sqrt{2\pi}}e^{-\frac{1}{2\nu^2}x^2}$. We can write
		\[
		\DD_\nu(\tau)(z) = \int_\RR \tau(z+\epsilon) G_\nu(\epsilon) d\epsilon.
		\]
		Then, using that $\int_\RR e^{i\epsilon} G_\nu(\epsilon) d\epsilon = e^{-\nu^2/2}$ and writing $\tau_a(z)$ as $1+2\epsilon\Re(e^{i(z-\theta)})$, we obtain
		
		\begin{eqnarray*}
			\DD_\nu(\tau_a)(z)  &=& \int_\RR \tau_a(z+\epsilon) G_\nu(\epsilon) d\epsilon     \\
			&=& \int_\RR \left(1+2\epsilon\Re\left(e^{i(z-\theta+\epsilon)}\right)\right) G_\nu(\epsilon) d\epsilon     \\
			&=& 1 + 2 \epsilon\Re\left(\int_\RR e^{i(z-\theta+\epsilon)} G_\nu(\epsilon) d\epsilon\right)     \\
			&=& 1 + 2 \epsilon\Re\left(e^{i(z-\theta)} \int_\RR e^{i\epsilon} G_\nu(\epsilon) d\epsilon\right)  \\
			&=& 1 + 2 \epsilon \Re\left(e^{i(z-\theta)}\right) e^{-\nu^2/2}      \\
			&=& \tau_{e^{-\nu^2/2}a}(z)                  \qedhere
		\end{eqnarray*}
	\end{proof}
	
	\begin{lemma*}
		Given $a \in \CC$ and $\nu\in \RR^+$ small, $\EE_{\tau_a} = a$, $\EE_{\DD_\nu(\tau_a)} = e^{-\nu^2/2} a$.
	\end{lemma*}
	This lemma is immediate from the two previous lemmas.
	
	\begin{lemma*}
		Given $a \in \CC$ and $P\in \CC$ ,
		\[
		\MM_P(\tau_a) \approx \tau_{a'},
		\]
		where $a' = P(P+4a)/8$.
	\end{lemma*}
	
	\begin{proof}%[Proof of Lemma \ref{L:Mp}]
		Let $a = (\epsilon/2) e^{i\delta}$, and $P = 2 m e^{i\theta}$, where the factors of $1/2$ and $2$ are used to simplify the notation below. By considering a rotation, we may assume $\theta = 0$ and $\delta= \delta-\theta$. This is evident since if $z' = z -\theta$, then $2 z + m \sin(\theta-z) = 2 z' + 2\theta + m \sin(-z')$. Thus, if we define $\tau'(z') = \tau(z'+\theta)$, then $\MM_P(\tau)(z'+2\theta) = \MM_{e^{-i\theta} P}(\tau')(z')$. We will use this to find the formula for $\theta\neq 0$ later. 
		
		We are applying the function $f(z) = 2z - 2 m\sin(z)$ to the distribution $\tau(z) = (1/2\pi)(1+\epsilon\cos(z - \delta))$. Thus, $\MM_P(\tau)(x) = \sum_{z\in f^{-1}(x)} \frac{\tau(z)}{f'(z)}$. For $m$ being small, the function $f$ is very close to $2z$, and, hence, it is 2-to-1 as a function on $S^1$. If $f(z) = x$, $z = x/2 + m\sin(z)$, which we can approximate to the first order (using $z\approx x/2$ and ignoring terms that contain $m^2$) by,  
		\[
		z = x/2 + m \sin(x/2)
		\]
		and
		\[
		2\pi\tau(z) = 1 + \epsilon \cos(x/2 + m \sin(x/2) - \delta),
		\]
		which, using the first-order Taylor expansion of  $cosine$ around $x/2-\delta$, we can approximate to the second order as
		\[
		2\pi\tau(z) = 1  + \epsilon\cos(x/2- \delta) - \epsilon m\sin(x/2) \sin(x/2 - \delta).
		\]
		
		The derivative of $f$ is given by $f'(z) = 2 - 2 m \cos(z)$, which we can approximate to  second order by 
		\[
		f'(z) \approx 2 - 2 m \cos(x/2 + m\sin(x/2)).
		\]
		Considering again the first-order Taylor expansion of $cosine$, we obtain
		\[
		f'(z)\approx 2 - 2m \cos(x/2) + 2m^2\sin^2(x/2).
		\]
		Using that the second degree approximation to the inverse of $1+am+bm^2$ is $1-am+(a^2- b) m^2$, we obtain
		\begin{eqnarray*}
			2/f'(z) &\approx& 1 + m \cos(x/2) + m^2(\cos^2(x/2)-\sin^2(x/2))            \\  
			&\approx& 1 + m \cos(x/2) + m^2\cos(x).
		\end{eqnarray*}
		
		There are two points in $f^{-1}(x)$, which we call $z_1$ and $z_2$. They correspond to the fact that $x$ and $x+2\pi$ represent the same element in $S^1$. So, we consider $z_1$ using $x$ as above and $z_2$ using $x+2\pi$. Then, considering that $\cos(x/2+\pi) = -\cos(x/2)$ and $\sin(x/2+\pi) = -\sin(x/2)$, we obtain
		\begin{eqnarray*}
			2\pi\tau(z_1) &=& 1 - \epsilon m\sin(x/2) \sin(x/2 - \delta) + \epsilon\cos(x/2- \delta)\\
			2\pi\tau(z_2) &=& 1 - \epsilon m\sin(x/2) \sin(x/2 - \delta) - \epsilon\cos(x/2- \delta) \\
			2/f'(z_1) &\approx& 1 + m^2\cos(x) + m \cos(x/2)       \\
			2/f'(z_2) &\approx& 1 + m^2\cos(x) - m \cos(x/2).
		\end{eqnarray*}
		
		Then, with the fact that $(A+B)(C+D)+(A-B)(C-D)= 2AC+2BD$, we can obtain
		\begin{widetext}
			\begin{eqnarray*}
				2\pi\MM_P(\tau)(x) &=&   2\pi \tau(z_1)/f'(z_1) + 2\pi \tau(z_2)/f'(z_2) \\
				&\approx&
				(1 - \epsilon m\sin(x/2) \sin(x/2 - \delta))(1 + m^2\cos(x)) + \epsilon m \cos(x/2- \delta)  \cos(x/2)         \\
				&\approx& 1 - \epsilon m\sin(x/2) \sin(x/2 - \delta) + m^2\cos(x) + \epsilon m\cos(x/2- \delta)\cos(x/2)          \\
				&=& 1 + m^2\cos(x) + \epsilon m\big(\cos(x/2- \delta)\cos(x/2) - \sin(x/2) \sin(x/2 - \delta)\big) \\
				&=& 1 + m^2\cos(x) + \epsilon m\cos(x-\delta) \\
				&=& 1 + \Re(e^{ix} (m^2 +  \epsilon me^{-i\delta}))          \\
				&=& 2\pi\tau_{b'}(x),
			\end{eqnarray*}
		\end{widetext}
		where $b' = (m^2 +  \epsilon me^{-i\delta})/2$.
		
		Now we remove the assumption that $\theta=0$. Recall that $\MM_P(\tau)(z'+2\theta) = \MM_{e^{-i\theta} P}(\tau')(z')$. We, thus, need to move our result by $2\theta$, i.e.,
		\begin{eqnarray*}
			2\pi\MM_P(\tau)(x)  &=&    1+\Re(e^{i(x-2\theta)}(m^2+\epsilon m e^{-i(\delta+\theta)}))   \\
			&=&            1 +\Re(e^{ix} \times m e^{-i\theta} \times (m e^{-i\theta} + \epsilon e^{-i \delta}))     \\
			&=&     2\pi      \tau_{a'}
		\end{eqnarray*}
		where $a' = (P/2)(P/2 + 2a)/2 = (P(P+4a))/8$.
	\end{proof}

	%%%%%%%%%%%%%%%%%%%%%%%%%%%%%%%%%%%%%%%%%%%%%%%%%%%%%%%%%%%%%%%%%%
	\section{The term $a_w$} \label{Sec:Appendix aw}
	In Sec. \ref{Sec:The main result}, the expression for the expected value of $V^+$ [Eq. \ref{eq: V+ formula}] presents the term
	\begin{equation}\label{eq: 2nd term}
		\int w^2 a_w e^{-\alpha^2 w/4C} \delta(w),
	\end{equation}
	which contains the randomly distributed values $a_w$. Since $a_w$ is the mean field of $\rho_w$, which is close to uniform before synchronization, we assume that, for each $w$, $a_w$ is normally distributed around $0$ with 2D-variance $1/(2 N\delta(w)dw)$. Notice that $N\delta(w)dw$ is the number of maps with connectivity between $w$ and $w+dw$. Notice also that
	\begin{equation}\label{eq: V=wawdelta}
		\int w a_w \delta(w) dw = V.
	\end{equation}
	Now, for each $w$, consider a collection of a random variable $Z_w$ in the complex plane, normally distributed with mean 0 and variance 1. Under these circumstances, we can assume that $a_w$s are described by $Z_w/\sqrt{2N\delta(w)}$. Hereupon, consider $a$ and $Z$ as random variables in $\RR^\RR$, with  $a= T(Z)$, where $T$ is the map from $\RR^\RR\to\RR^\RR$ given by $T(g)(w) = g(w)/\sqrt{2N\delta(w)}$. Let $C,D\in\RR^\RR$ be the maps,
	\begin{itemize}
		\item $C(w) = w\delta(w)$,
		\item $D(w)=w^2 e^{-w/4C} \delta(w).$
	\end{itemize}
	The motivation to introduce these maps is to write Eq. \ref{eq: 2nd term} as
	\[
	\langle D, a\rangle,
	\]
	where the inner product is defined as $\langle g,f\rangle=\int f(w)g(w)dw$. This formulation also allows us to rewrite Eq. \ref{eq: V=wawdelta} as,
	\[
	\langle C, a\rangle = V.
	\]
	Notice that $\langle f, T(g)\rangle = \langle T(f), g\rangle$ for all $f,g\in\RR^\RR$.
	Therefore, we can procedure as follows: We want to obtain the expected value of $\langle T(D), Z\rangle$, given that $\langle T(C), Z\rangle = V$. To this end, we decompose $T(D)$ into orthogonal functions as
	\[
	T(D)=\beta T(C) + E
	\]
	where $T(C)$ and $E$ are orthogonal (i.e., $\langle T(C),E\rangle=0)$. By doing so, we obtain
	\[
	\langle D, a\rangle = \langle T(D), Z\rangle = \beta V + \langle E,Z\rangle.
	\]
	Thus, $\langle D, a\rangle$ has an expected value $\beta V$ and a variance $||E||^2$.
	We calculate $\beta$ and $E$ as follows:
	\[
	\beta = \frac{\langle T(C), T(D)\rangle}{||T(C)||^2}
	\mbox{,}\quad
	E = T(D) - \beta T(C).
	\]
	The numerator and denominator of $\beta$, respectively, are then obtained as
	\begin{eqnarray*}
		\langle T(C), T(D)\rangle &=& \int \frac{w\delta(w)}{\sqrt{2N\delta(w)}}\frac{w^2e^{-\alpha^2 w/4C}\delta(w)}{\sqrt{2N\delta(w)}}dw      \\
		&=& \frac{1}{2N}\int w^3e^{-\alpha^2 w/4C}\delta(w)dw      \\
		&=& \frac{1}{2N} M_3(\delta)\ K^\alpha_\delta(C),
	\end{eqnarray*}
	\begin{eqnarray*}
		\langle T(C), T(C)\rangle &=& \int \frac{w\delta(w)}{\sqrt{2N\delta(w)}}\frac{w\delta(w)}{\sqrt{2N\delta(w)}}dw      \\
		&=& \frac{1}{2N}\int w^2\delta(w)dw      \\
		&=& \frac{1}{2N} M_2.
	\end{eqnarray*}
	
	We conclude that $\int w^2 a_w e^{-\alpha^2 w/4C} \delta(w)$ has mean $V M_3(\delta)K^\alpha_\delta(C)/M_2(\delta)$, as needed in the derivation of Sec. \ref{Sec:The continuous approximation}. As for $||E||^2$, we obtain
	\[
	\langle E, E\rangle = \frac{1}{2N} \int \big(w^2e^{-\alpha^2 w/4C} -w(K^\alpha_\delta(C)/M_2)\big)^2 \delta(w) dw.
	\]
	This value is rather negligible in all examples we considered in this work. For all 81 settings considered in Fig. \ref{Fig:Figure2}, the mean value for $\alpha ||E||$ was 0.03 with maximum 0.07. Putting the results above with Eq. \ref{eq: V+ formula}, we obtain that the randomness of $a_w$s adds 2D-noise with standard deviation $||V||\alpha ||E||$. Therefore, it is only a small percentage of the size of $||V||$ that is being added as noise. The effect is probably not completely negligible, and more investigation is suggested for future works in this area.

	%%%%%%%%%%%%%%%%%%%%%%%%%%%%%%%%%%%%%%%%%%%%%%%%%%%%%%%%%%%%%%%%%%%%%%%%%%%%%%%%%%%%%%%%%%%%%%%%%%%%%%%%%%%%%%%%%%%%%%%%%%%%%%%%%%%%%%%%%%%%%%%%%%%%%%%%%%%%%%%%%%%%%%%%%%%%%%%%%%%%%%%%%%%%%%%%%%%%%%%%%%%%%%%%%%%%%%%%%%%%%%%%%%%%%%%%%%%%%%%%%%%%%%%%%%%%%%%%%%%%%%%%%%%%%%%%%%%%%%%%%%%%%%%%%%%%%%%%%%%%%%%%%%%%%%%%%%%%%%%%%%%%%%%%%%%%%%%%%%%%%%%%%%%%%%%%%%%%%%%%%%%%%%%%%%%%%%%%%%%%%%%%%%%%%%%%%%%%%%%%%%%%%%%%%%%%%%%%%%%%%%%%%%%%%%%%%%%%%%%%%%%%%%%%%%%%%%%%%%%%%%%%%%%%%%%%%%%%%%%%%%%%%%%%%%%%%%%%%%%%%%%%%%%%%%%%%%%%%%%%%%%%%%%%%%%%%%%%%%%%%%%%%%%%%%%%%%%%%%%%%%%%%%%%%%%%%%%%%%%%%%%%%%%%%%%%%%%%%%%%%%%%%%%%%%%%%%%%%%%%%%%%%%%%%%%%%%%%%%%%%%%%%%%%%%%%%%%%%%%%%%%%%%%%%%%%%%%%%%%%%%%%%%%%%%%%%%%%%%%%%%%%%%%%%%%%%%%%%%%%%%%%%%%%%%%%%%%%%%%%%%%%%%%%%%%%%%%%%%%%%%%%%%%%%%%%%%%%%%%%%%%%%%%%%%%%%%%%%%%%%%%%%%%%%%%%%%%%%%%%%%%%%%%%%%%%%%%%%%%%%%%%%%%%%%%%%%%%%%%%%%%%%%%%%%%%%%%%%%%%%%%%%%%%%%%%%%%%%%%%%%%%%%%%%%
	% References
	\bibliography{references}
	
\end{document}